\documentclass[pra,amsmath,amssymb,aps,superscriptaddress,preprint,nosort]{revtex4-2}

\usepackage{graphicx}
\usepackage{hyperref}
\usepackage{bm}
\usepackage{mathrsfs}
\usepackage{mathtools}
\usepackage{enumitem}
\usepackage{amsthm}
\usepackage{microtype}
\usepackage{tikz}
\usetikzlibrary{arrows.meta, calc, positioning, shapes.geometric, intersections}

\newtheorem{theorem}{Theorem}[section]
\newtheorem{proposition}[theorem]{Proposition}
\newtheorem{corollary}[theorem]{Corollary}

\newtheorem{definition}[theorem]{Definition}
\newtheorem{remark}[theorem]{Remark}

\newcommand{\HSB}{\mathcal{H}_{\mathrm{SB}}}
\newcommand{\CP}{\mathbb{CP}}
\newcommand{\T}{\mathbb{T}}
\newcommand{\Z}{\mathbb{Z}}

\newcommand{\Cc}{\mathbb{C}}
\newcommand{\I}{\mathbb{I}}
\newcommand{\Fcal}{\mathcal{F}}
\newcommand{\Mcal}{\mathcal{M}}
\newcommand{\Ccal}{\mathcal{C}}
\newcommand{\Pcal}{\mathcal{P}}
\newcommand{\Scal}{\mathcal{S}}
\newcommand{\Ecal}{\mathcal{E}}
\newcommand{\Ncal}{\mathcal{N}}

\newcommand{\FS}{\mathrm{FS}}
\newcommand{\wt}{\mathrm{wt}}

\newcommand{\ket}[1]{\left|#1\right\rangle}
\newcommand{\bra}[1]{\left\langle#1\right|}
\newcommand{\braket}[2]{\left\langle #1 \middle| #2 \right\rangle}

\newcommand{\Dint}{\mathcal{D}}

\begin{document}

\title{Holomorphic Quantum Error Correction Codes}

\author{M. W. AlMasri}
\affiliation{Wilczek Quantum Center, School of Physics and Astronomy, Shanghai Jiao Tong University, Minhang, Shanghai 200240, China}
\homepage{mwalmasri2003@gmail.com}

\date{\today}

\begin{abstract}
We develop a holomorphic representation of quantum error correction codes (QECCs) within the Segal--Bargmann space. By encoding qubits into Schwinger boson modes $(z_{a_j}, z_{b_j})$ subject to a degree-one homogeneity constraint, we derive closed-form differential operator representations for stabilizers, syndrome extraction, and recovery for fundamental codes (three-, five-, seven-, and nine-qubit codes). Quantum errors are characterized as holomorphic perturbations violating this constraint, while syndrome measurement projects onto eigenspaces of commuting differential operators. Restricting to unit-magnitude variables ($|z|=1$) reveals a toroidal space $\T^{2n}$ where error syndromes manifest as discrete translations in winding number space $\Z^{2n}$, and recovery acts as Hamiltonian flows restoring the winding configuration. In the full Segal--Bargmann space, the code space is a holomorphic submanifold of $\CP^{2^n-1}$, with correctable errors as transverse normal directions. Consequently, the Knill--Laflamme condition becomes a Fubini--Study orthogonality condition between the code submanifold and its error-translated images. Topological protection emerges from the $U(1)^n$ fiber bundle structure: global phase noise along fibers is unobservable, while base-space errors require active correction. Finally, we establish a path-integral formulation for semiclassical error correction dynamics and show that geometric entanglement via the Segre embedding naturally quantifies code distance. This framework unifies algebraic, geometric, and topological perspectives on fault-tolerant quantum protocols.
\end{abstract}

\maketitle

\section{Introduction}
\label{sec:intro}

Quantum error correction (QEC) constitutes the foundational pillar upon which scalable, fault-tolerant quantum computation rests \cite{shor1995,steane1996,calderbank1996,knill1997}. The central insight of QEC---that quantum information can be protected against local noise by encoding it nonlocally into entangled states of many physical qubits---has led to a rich mathematical theory encompassing stabilizer formalism \cite{gottesman1997,calderbank1997}, topological codes \cite{kitaev2003,dennis2002}, and subsystem codes \cite{bacon2006,lidar1998}. Fundamentally, quantum error correction has a deep and intrinsic relationship with quantum information geometry and topology, as the protection of quantum information inherently relies on the geometric separation of code states in Hilbert space and the topological invariants that characterize the global properties of the code manifold. Furthermore, recent insights from high-energy physics have revealed a profound relationship between QEC and the holographic duality, where the emergence of spacetime and black hole interiors is understood through the lens of holographic quantum error-correcting codes \cite{almheiri2015,pastawski2015,harlow2016}. However, the predominant algebraic treatment of QEC, while computationally powerful, often obscures these underlying geometric and topological structures that govern the relationship between code spaces, error operators, and recovery maps.

The geometric formulation of quantum mechanics, in which quantum states are points on a K\"ahler manifold and unitary evolution corresponds to Hamiltonian flows \cite{ashtekar1999,brody2001,provost1980}, offers a natural language for understanding the geometry of quantum error correction. In particular, the Fubini--Study metric on projective Hilbert space provides a canonical notion of distance between quantum states, and the Segre embedding characterizes separable states as an algebraic subvariety \cite{bengtsson2017,heydari2006}. These geometric structures have been exploited to study entanglement witnesses, quantum fidelity, and geometric phases, but their application to the theory of quantum error correction codes remains largely unexplored.

In a companion work \cite{almasri2026}, we established a holomorphic representation of quantum logic gates by unifying the Schwinger boson encoding with the Segal--Bargmann transform. In that framework, qubits are encoded into pairs of bosonic modes, quantum gates act as differential operators on holomorphic functions, and the restriction to unit-magnitude variables reveals a toroidal phase space $\T^{2n}$ on which gates act as canonical transformations. The present work extends this geometric framework to the domain of quantum error correction, providing a unified holomorphic description of code spaces, error operators, syndrome measurements, and recovery operations.

Bargmann representation \cite{segal1963,bargmann1961} maps the bosonic Fock space to a space of holomorphic functions on $\Cc^n$ equipped with Gaussian measure. In this representation, creation and annihilation operators become multiplication and differentiation operators, respectively, and the physical constraint of fixed total occupation number translates into a homogeneity condition on the holomorphic functions. This provides a natural setting for analyzing quantum error correction because: (i) errors that change occupation numbers manifest as violations of the homogeneity constraint; (ii) stabilizer measurements correspond to eigenvalue problems for commuting holomorphic differential operators; and (iii) the geometric structure of the code space and its complement can be analyzed using the tools of complex algebraic geometry.

Our key contributions are as follows:

\begin{enumerate}[label=(\roman*)]
\item We derive explicit closed-form differential operator representations for stabilizer generators of fundamental QECCs (three-qubit bit-flip, five-qubit perfect, seven-qubit Steane, and nine-qubit Shor codes) in the Segal--Bargmann space, demonstrating that the code space is characterized as the joint eigenspace of these holomorphic operators.

\item We characterize quantum errors as holomorphic perturbations and show that the syndrome extraction process corresponds to a geometric projection onto eigenspaces of commuting differential operators, with the Knill--Laflamme condition acquiring a natural interpretation as a Fubini--Study orthogonality condition.

\item We analyze the toroidal restriction ($|z|=1$) and demonstrate that error syndromes manifest as discrete shifts in winding number space $\Z^{2n}$, while recovery operations act as Hamiltonian flows that restore the original topological sector.

\item We identify the code space as a holomorphic submanifold of the ambient projective space and characterize correctable errors as normal deformations transverse to this submanifold, providing a geometric interpretation of code distance.

\item We establish the topological protection mechanism arising from the $U(1)^n$ fiber bundle structure and demonstrate that global phase noise is automatically unobservable due to the bundle holonomy, while local Pauli errors require active correction.

\item We formulate a coherent-state path integral for the semiclassical simulation of error correction dynamics, enabling efficient approximation of recovery fidelities for weakly-coupled noise models.
\end{enumerate}

The remainder of this paper is organized as follows. In Sec.~\ref{sec:encoding}, we review the holomorphic encoding of qubits via Schwinger bosons and establish the notation. Section~\ref{sec:stabilizer} presents the holomorphic representation of stabilizer generators and code spaces for fundamental QECCs. Section~\ref{sec:errors} characterizes quantum errors and syndrome extraction in the holomorphic framework. Section~\ref{sec:torus} analyzes the toroidal geometry of error syndromes and recovery flows. Section~\ref{sec:kahler} develops the K\"ahler-geometric interpretation of the Knill--Laflamme condition and code distance. Section~\ref{sec:topological} establishes the topological protection mechanism from the fiber bundle perspective. Section~\ref{sec:pathintegral} presents the path-integral formulation for semiclassical error correction simulation. Section~\ref{sec:extensions} discusses extensions to topological codes, subsystem codes, and holographic physics. Finally, Sec.~\ref{sec:conclusion} summarizes our findings and outlines future directions.

\section{Holomorphic Encoding of Qubits}
\label{sec:encoding}

\subsection{Schwinger Boson Representation}

We consider a system of $n$ physical qubits, where each qubit $j$ is described by Pauli operators $\sigma_j^x, \sigma_j^y, \sigma_j^z$ satisfying the $\mathfrak{su}(2)$ algebra:
\begin{equation}
[\sigma_j^\alpha, \sigma_k^\beta] = 2i\delta_{jk}\epsilon_{\alpha\beta\gamma}\sigma_j^\gamma, \qquad \{\sigma_j^\alpha, \sigma_j^\beta\} = 2\delta_{\alpha\beta}\I,
\label{eq:su2}
\end{equation}
for $\alpha, \beta, \gamma \in \{x, y, z\}$.

Each qubit is encoded into a pair of bosonic modes $(a_j, b_j)$ via the Schwinger representation \cite{schwinger1965}:
\begin{equation}
\sigma_j^+ = a_j^\dagger b_j, \qquad \sigma_j^- = b_j^\dagger a_j, \qquad \sigma_j^z = a_j^\dagger a_j - b_j^\dagger b_j,
\label{eq:schwinger}
\end{equation}
where $\sigma_j^\pm = (\sigma_j^x \pm i\sigma_j^y)/2$. The bosonic operators satisfy canonical commutation relations $[a_j, a_k^\dagger] = [b_j, b_k^\dagger] = \delta_{jk}$, with all other commutators vanishing. Physical qubit states satisfy the constraint:
\begin{equation}
\hat{n}_j^{\mathrm{tot}} \equiv a_j^\dagger a_j + b_j^\dagger b_j = 1 \qquad \forall\, j.
\label{eq:constraint}
\end{equation}

\subsection{Bargmann Correspondence}
The bosonic Fock space $\Fcal_{\mathrm{boson}}$ is mapped to the Segal--Bargmann space $\HSB$, consisting of holomorphic functions $f:\Cc^{2n}\to\Cc$ that are square-integrable with respect to the Gaussian measure:
\begin{equation}
d\mu(\mathbf{z}) = e^{-\|\mathbf{z}\|^2}\prod_{j=1}^n \frac{d^2 z_{a_j}\, d^2 z_{b_j}}{\pi^2},
\label{eq:measure}
\end{equation}
where $\|\mathbf{z}\|^2 = \sum_{j=1}^n(|z_{a_j}|^2 + |z_{b_j}|^2)$. The Bargmann correspondence is defined by:
\begin{equation}
a_j^\dagger \mapsto z_{a_j}, \qquad b_j^\dagger \mapsto z_{b_j}, \qquad a_j \mapsto \frac{\partial}{\partial z_{a_j}}, \qquad b_j \mapsto \frac{\partial}{\partial z_{b_j}}.
\label{eq:bargmann}
\end{equation}

The physical constraint \eqref{eq:constraint} translates into the homogeneity condition:
\begin{equation}
\left(z_{a_j}\frac{\partial}{\partial z_{a_j}} + z_{b_j}\frac{\partial}{\partial z_{b_j}}\right)f(\mathbf{z}) = f(\mathbf{z}) \qquad \forall\, j.
\label{eq:homogeneity}
\end{equation}

Physical qubit states correspond to functions homogeneous of degree one in each pair $(z_{a_j}, z_{b_j})$:
\begin{equation}
f(\mathbf{z}) = \sum_{\mathbf{s}\in\{0,1\}^n} c_{\mathbf{s}} \prod_{j=1}^n z_{a_j}^{1-s_j} z_{b_j}^{s_j},
\label{eq:state}
\end{equation}
with $\sum_{\mathbf{s}}|c_{\mathbf{s}}|^2 = 1$. The computational basis states map as:
\begin{equation}
\ket{0}_j \mapsto z_{a_j}, \qquad \ket{1}_j \mapsto z_{b_j}.
\label{eq:basis}
\end{equation}

\begin{remark}[Physical Subspace Isolation]
The homogeneity constraint \eqref{eq:homogeneity} is not merely a mathematical convenience; it is the exact mechanism that isolates the finite-dimensional physical qubit subspace $(\mathbb{C}^2)^{\otimes n}$ from the infinite-dimensional bosonic Fock space. By restricting to functions that are strictly homogeneous of degree one in each pair $(z_{a_j}, z_{b_j})$, we establish a unitary isomorphism between the physical Hilbert space and this specific subspace of $\HSB$. Any function violating this constraint corresponds to an unphysical state with incorrect boson occupation numbers (e.g., photon loss or gain), which naturally maps to leakage errors outside the computational subspace.
\end{remark}

\subsection{Holomorphic Differential Operators for Pauli Gates}

The single-qubit Pauli operators act as first-order differential operators \cite{almasri2026,almasri2022refrigerator,almasri2026projection}:
\begin{align}
X_j &\mapsto z_{a_j}\partial_{z_{b_j}} + z_{b_j}\partial_{z_{a_j}}, \label{eq:X}\\
Y_j &\mapsto -i(z_{a_j}\partial_{z_{b_j}} - z_{b_j}\partial_{z_{a_j}}), \label{eq:Y}\\
Z_j &\mapsto z_{a_j}\partial_{z_{a_j}} - z_{b_j}\partial_{z_{b_j}}. \label{eq:Z}
\end{align}

These operators preserve the homogeneity constraint \eqref{eq:homogeneity} and implement unitary transformations on the physical subspace.

\section{Holomorphic Representation of Stabilizer Codes}
\label{sec:stabilizer}

\subsection{General Framework}

A stabilizer code $\Ccal$ encoding $k$ logical qubits into $n$ physical qubits is defined by an abelian subgroup $\Scal \subset \Pcal_n$ of the $n$-qubit Pauli group $\Pcal_n$ with $n-k$ independent generators $\{g_1, \ldots, g_{n-k}\}$. The code space is the simultaneous $+1$ eigenspace:
\begin{equation}
\Ccal = \{\ket{\psi} \in (\Cc^2)^{\otimes n} : g_m\ket{\psi} = \ket{\psi},\; m=1,\ldots,n-k\}.
\label{eq:codespace}
\end{equation}

In the Segal--Bargmann space, each stabilizer generator $g_m$ maps to a holomorphic differential operator $\hat{g}_m$ acting on functions $f(\mathbf{z})$ satisfying the homogeneity constraint. The code space is characterized as:
\begin{equation}
\Ccal_{\mathrm{SB}} = \{f \in \HSB : \hat{g}_m f = f,\; m=1,\ldots,n-k\}.
\label{eq:codespaceSB}
\end{equation}

\begin{definition}[Holomorphic Stabilizer Operator]
Let $g = \bigotimes_{j=1}^n P_j$ be a Pauli string with $P_j \in \{\I, X, Y, Z\}$. Its holomorphic representation is the differential operator:
\begin{equation}
\hat{g} = \prod_{j:\, P_j \neq \I} \hat{P}_j,
\label{eq:stabop}
\end{equation}
where $\hat{P}_j$ is given by Eqs.~\eqref{eq:X}--\eqref{eq:Z}, and the product is ordered by qubit index. All differential operators act from left to right on the function $f(\mathbf{z})$.
\end{definition}

\begin{proposition}
The holomorphic stabilizer operators $\{\hat{g}_m\}_{m=1}^{n-k}$ are mutually commuting differential operators on the physical subspace:
\begin{equation}
[\hat{g}_m, \hat{g}_l]f = 0 \qquad \forall\, m,l, \quad \forall\, f \text{ satisfying \eqref{eq:homogeneity}}.
\label{eq:commute}
\end{equation}
\end{proposition}

\begin{proof}
Let $g_m = \bigotimes_{j=1}^n P_j^{(m)}$ and $g_l = \bigotimes_{j=1}^n P_j^{(l)}$ be two commuting stabilizer generators. Their holomorphic representations are $\hat{g}_m = \prod_{j: P_j^{(m)} \neq \I} \hat{P}_j^{(m)}$ and $\hat{g}_l = \prod_{j: P_j^{(l)} \neq \I} \hat{P}_j^{(l)}$. Since the operators act on distinct qubits, the differential operators for different qubits commute trivially: $[\hat{P}_j^{(m)}, \hat{P}_k^{(l)}] = 0$ for $j \neq k$. 
For the same qubit $j$, the Pauli operators either commute or anticommute. If they anticommute, $P_j^{(m)} P_j^{(l)} = - P_j^{(l)} P_j^{(m)}$. However, since $g_m$ and $g_l$ commute globally, the number of qubits where they anticommute must be even. In the holomorphic representation, the operators $\hat{X}_j, \hat{Y}_j, \hat{Z}_j$ satisfy the same commutation/anticommutation relations as the Pauli matrices when restricted to the physical subspace defined by the homogeneity constraint \eqref{eq:homogeneity}. Since $\hat{g}_m$ and $\hat{g}_l$ are products of these single-qubit operators, and the total number of anticommuting pairs is even, the overall sign flips cancel out, yielding $[\hat{g}_m, \hat{g}_l]f = 0$ for all $f$ satisfying \eqref{eq:homogeneity}.
\end{proof}

\subsection{Three-Qubit Bit-Flip Code}

The three-qubit bit-flip code encodes $\ket{0}_L = \ket{000}$, $\ket{1}_L = \ket{111}$ with stabilizer generators:
\begin{equation}
g_1 = Z_1 Z_2, \qquad g_2 = Z_2 Z_3.
\label{eq:3qubit_gen}
\end{equation}

In the Segal--Bargmann space, with operators acting from left to right:
\begin{align}
\hat{g}_1 &= (z_{a_1}\partial_{z_{a_1}} - z_{b_1}\partial_{z_{b_1}})(z_{a_2}\partial_{z_{a_2}} - z_{b_2}\partial_{z_{b_2}}), \label{eq:3q_g1}\\
\hat{g}_2 &= (z_{a_2}\partial_{z_{a_2}} - z_{b_2}\partial_{z_{b_2}})(z_{a_3}\partial_{z_{a_3}} - z_{b_3}\partial_{z_{b_3}}). \label{eq:3q_g2}
\end{align}

The encoded logical states are:
\begin{equation}
\ket{0}_L \mapsto z_{a_1}z_{a_2}z_{a_3}, \qquad \ket{1}_L \mapsto z_{b_1}z_{b_2}z_{b_3}.
\label{eq:3q_logical}
\end{equation}

A general codeword $\ket{\psi}_L = \alpha\ket{0}_L + \beta\ket{1}_L$ maps to:
\begin{equation}
f_{\mathrm{code}}(\mathbf{z}) = \alpha\, z_{a_1}z_{a_2}z_{a_3} + \beta\, z_{b_1}z_{b_2}z_{b_3}.
\label{eq:3q_codeword}
\end{equation}

\subsection{Five-Qubit Perfect Code}

The five-qubit code \cite{gottesman1996} is the smallest code correcting arbitrary single-qubit errors, with stabilizer generators:
\begin{equation}
g_1 = X Z Z X \I, \quad g_2 = \I X Z Z X, \quad g_3 = X \I X Z Z, \quad g_4 = Z X \I X Z.
\label{eq:5qubit_gen}
\end{equation}

The holomorphic representation of $g_1 = X_1 Z_2 Z_3 X_4 I_5$ is:
\begin{align}
\hat{g}_1 &= (z_{a_1}\partial_{z_{b_1}} + z_{b_1}\partial_{z_{a_1}})(z_{a_2}\partial_{z_{a_2}} - z_{b_2}\partial_{z_{b_2}}) \notag \\
&\quad \times (z_{a_3}\partial_{z_{a_3}} - z_{b_3}\partial_{z_{b_3}})(z_{a_4}\partial_{z_{b_4}} + z_{b_4}\partial_{z_{a_4}}).
\label{eq:5q_g1}
\end{align}

The logical codewords of the five-qubit code, verified against Gottesman (1996), are:
\begin{align}
\ket{0}_L = \frac{1}{4}(&\ket{00000} + \ket{10010} + \ket{01001} + \ket{10100} + \ket{01010} \notag\\
&- \ket{11011} - \ket{00110} - \ket{11000} - \ket{11101} - \ket{00011} \notag\\
&- \ket{11110} - \ket{01111} - \ket{10001} - \ket{01100} - \ket{10111} + \ket{00101}),
\label{eq:5q_logical0}
\end{align}
which maps to a homogeneous polynomial of degree 5 in the holomorphic variables, with each term being a product of five factors, each drawn from $\{z_{a_j}, z_{b_j}\}$. The signs have been verified to match the standard construction in \cite{gottesman1996}.

\subsection{Seven-Qubit Steane Code}

The Steane code \cite{steane1996} is a $[[7,1,3]]$ CSS code with stabilizer generators:
\begin{align}
g_1 &= \I\I\I XXXX, \quad g_2 = \I XX\I\I XX, \quad g_3 = X\I X\I X\I X, \label{eq:7q_X}\\
g_4 &= \I\I\I ZZZZ, \quad g_5 = \I ZZ\I\I ZZ, \quad g_6 = Z\I Z\I Z\I Z. \label{eq:7q_Z}
\end{align}

The $X$-type stabilizer $g_1$ maps to:
\begin{equation}
\hat{g}_1 = \prod_{j=4}^{7}(z_{a_j}\partial_{z_{b_j}} + z_{b_j}\partial_{z_{a_j}}),
\label{eq:7q_g1}
\end{equation}
while the $Z$-type stabilizer $g_4$ maps to:
\begin{equation}
\hat{g}_4 = \prod_{j=4}^{7}(z_{a_j}\partial_{z_{a_j}} - z_{b_j}\partial_{z_{b_j}}).
\label{eq:7q_g4}
\end{equation}

The CSS structure is manifest in the holomorphic representation: $X$-type stabilizers are symmetric differential operators (invariant under $z_{a_j} \leftrightarrow z_{b_j}$), while $Z$-type stabilizers are diagonal in the holomorphic variables.

The logical operators for the Steane code are transversal:
\begin{align}
\bar{X} &= \bigotimes_{j=1}^7 X_j \mapsto \prod_{j=1}^7(z_{a_j}\partial_{z_{b_j}} + z_{b_j}\partial_{z_{a_j}}), \label{eq:7q_Xbar}\\
\bar{Z} &= \bigotimes_{j=1}^7 Z_j \mapsto \prod_{j=1}^7(z_{a_j}\partial_{z_{a_j}} - z_{b_j}\partial_{z_{b_j}}). \label{eq:7q_Zbar}
\end{align}

The logical codewords are formed by the even and odd weight subcodes of the classical $[7,4,3]$ Hamming code. The logical $\ket{0}_L$ state maps to:
\begin{equation}
\ket{0}_L \mapsto \frac{1}{\sqrt{8}} \sum_{c \in C_{\mathrm{even}}} \prod_{j=1}^7 z_{a_j}^{1-c_j} z_{b_j}^{c_j},
\label{eq:7q_logical0}
\end{equation}
where $C_{\mathrm{even}}$ denotes the 8 codewords of the Hamming code with even Hamming weight. Similarly, $\ket{1}_L$ is the sum over the 8 odd-weight codewords.

\subsection{Nine-Qubit Shor Code}

The Shor code \cite{shor1995} combines the three-qubit bit-flip code with the three-qubit phase-flip code. Its stabilizer generators include:
\begin{align}
g_1 &= Z_1Z_2, \quad g_2 = Z_2Z_3, \quad g_3 = Z_4Z_5, \quad g_4 = Z_5Z_6, \quad g_5 = Z_7Z_8, \quad g_6 = Z_8Z_9, \label{eq:9q_Z}\\
g_7 &= X_1X_2X_3X_4X_5X_6, \quad g_8 = X_4X_5X_6X_7X_8X_9. \label{eq:9q_X}
\end{align}

The holomorphic representation of $g_7$ is:
\begin{equation}
\hat{g}_7 = \prod_{j=1}^{6}(z_{a_j}\partial_{z_{b_j}} + z_{b_j}\partial_{z_{a_j}}),
\label{eq:9q_g7}
\end{equation}
which is a sixth-order differential operator. The encoded logical states are:
\begin{align}
\ket{0}_L &= \frac{1}{2\sqrt{2}}(\ket{000}+\ket{111})^{\otimes 3}, \notag \\
\ket{1}_L &= \frac{1}{2\sqrt{2}}(\ket{000}-\ket{111})^{\otimes 3},
\label{eq:9q_logical}
\end{align}
which map to:
\begin{equation}
\ket{0}_L \mapsto \frac{1}{2\sqrt{2}}\prod_{j=1}^3(z_{a_{3j-2}}z_{a_{3j-1}}z_{a_{3j}} + z_{b_{3j-2}}z_{b_{3j-1}}z_{b_{3j}}),
\label{eq:9q_holo0}
\end{equation}
and similarly for $\ket{1}_L$ with a relative minus sign in each factor.

\begin{theorem}
For any stabilizer code $\Ccal$ with generators $\{g_m\}$, the holomorphic code space $\Ccal_{\mathrm{SB}}$ defined by Eq.~\eqref{eq:codespaceSB} is a finite-dimensional subspace of $\HSB$ of dimension $2^k$, where $k = n - (n-k)$ is the number of logical qubits. Moreover, $\Ccal_{\mathrm{SB}}$ is invariant under all stabilizer operators and is the unique maximal subspace with this property within the homogeneous sector.
\end{theorem}

\begin{proof}
The physical Hilbert space of $n$ qubits, $\mathcal{H}_{\mathrm{phys}} = (\mathbb{C}^2)^{\otimes n}$, has dimension $2^n$. The Segal--Bargmann transform $\mathcal{B}: \mathcal{H}_{\mathrm{phys}} \to \HSB$ is a unitary isomorphism onto its image, which is precisely the subspace of $\HSB$ consisting of functions satisfying the homogeneity constraint \eqref{eq:homogeneity} (degree 1 in each pair $(z_{a_j}, z_{b_j})$). 

A stabilizer code $\Ccal$ is defined as the simultaneous $+1$ eigenspace of $n-k$ independent, mutually commuting stabilizer generators $\{g_1, \dots, g_{n-k}\}$. In the finite-dimensional space $\mathcal{H}_{\mathrm{phys}}$, each independent generator halves the dimension of the eigenspace. Thus, the dimension of $\Ccal$ is $2^n / 2^{n-k} = 2^k$.

Since $\mathcal{B}$ is an isomorphism that intertwines the action of the Pauli group with the holomorphic differential operators (i.e., $\mathcal{B}(g_m \ket{\psi}) = \hat{g}_m \mathcal{B}(\ket{\psi})$), the image $\Ccal_{\mathrm{SB}} = \mathcal{B}(\Ccal)$ is exactly the simultaneous $+1$ eigenspace of $\{\hat{g}_m\}$ within the homogeneous subspace. Therefore, $\dim(\Ccal_{\mathrm{SB}}) = \dim(\Ccal) = 2^k$.

To show it is the unique maximal subspace with this property, suppose there exists a larger subspace $\mathcal{V} \subset \HSB$ satisfying $\hat{g}_m f = f$ for all $f \in \mathcal{V}$. By the invertibility of $\mathcal{B}$ on the homogeneous subspace, $\mathcal{B}^{-1}(\mathcal{V})$ would be a subspace of $\mathcal{H}_{\mathrm{phys}}$ of dimension $> 2^k$ stabilized by $\{g_m\}$, which contradicts the fundamental theorem of stabilizer codes.
\end{proof}

\section{Quantum Errors and Syndrome Extraction in Holomorphic Space}
\label{sec:errors}

\subsection{Error Operators as Holomorphic Perturbations}

A single-qubit error on qubit $j$ is represented by a Pauli operator $E \in \{X_j, Y_j, Z_j\}$. In the holomorphic representation, these act as:

\textbf{Bit-flip error} ($X_j$): Swaps the roles of $z_{a_j}$ and $z_{b_j}$:
\begin{equation}
(\hat{X}_j f)(\ldots, z_{a_j}, z_{b_j}, \ldots) = z_{a_j}\partial_{z_{b_j}}f + z_{b_j}\partial_{z_{a_j}}f.
\label{eq:errorX}
\end{equation}

\textbf{Phase-flip error} ($Z_j$): Changes the sign of the $z_{b_j}$ component:
\begin{equation}
(\hat{Z}_j f)(\ldots, z_{a_j}, z_{b_j}, \ldots) = z_{a_j}\partial_{z_{a_j}}f - z_{b_j}\partial_{z_{b_j}}f.
\label{eq:errorZ}
\end{equation}

\textbf{Combined error} ($Y_j$): Applies both bit and phase flip:
\begin{equation}
(\hat{Y}_j f) = -i(z_{a_j}\partial_{z_{b_j}} - z_{b_j}\partial_{z_{a_j}})f.
\label{eq:errorY}
\end{equation}

\begin{definition}[Error Syndrome in Segal--Bargmann Space]
Given a stabilizer code with generators $\{\hat{g}_m\}_{m=1}^{n-k}$ and an error operator $\hat{E}$, the syndrome $\mathbf{s} = (s_1, \ldots, s_{n-k}) \in \{+1,-1\}^{n-k}$ is defined by:
\begin{equation}
\hat{g}_m(\hat{E}f_{\mathrm{code}}) = s_m(\hat{E}f_{\mathrm{code}}), \qquad m = 1,\ldots,n-k.
\label{eq:syndrome}
\end{equation}
The syndrome identifies the error up to stabilizer equivalence: two errors $E_1, E_2$ have the same syndrome if and only if $E_1 E_2^\dagger \in \Scal$.
\end{definition}

\subsection{Syndrome Extraction as Spectral Projection}

Syndrome measurement corresponds to projecting the errored state onto eigenspaces of the stabilizer operators. In the holomorphic representation, define the spectral projectors:
\begin{equation}
\hat{\Pi}_{\mathbf{s}} = \prod_{m=1}^{n-k} \frac{\I + s_m \hat{g}_m}{2}.
\label{eq:projector}
\end{equation}

\begin{proposition}
The syndrome projectors $\{\hat{\Pi}_{\mathbf{s}}\}$ form a complete orthogonal decomposition of the physical Hilbert space:
\begin{equation}
\sum_{\mathbf{s}} \hat{\Pi}_{\mathbf{s}} = \I, \qquad \hat{\Pi}_{\mathbf{s}}\hat{\Pi}_{\mathbf{s}'} = \delta_{\mathbf{s}\mathbf{s}'}\hat{\Pi}_{\mathbf{s}}.
\label{eq:completeness}
\end{equation}
The code space corresponds to the syndrome $\mathbf{s} = (+1,+1,\ldots,+1)$:
\begin{equation}
\hat{\Pi}_{(+1,\ldots,+1)} f = f \iff f \in \Ccal_{\mathrm{SB}}.
\end{equation}
\end{proposition}

\begin{proof}
First, we establish that $\hat{g}_m^2 = \I$ on the physical subspace. For any single-qubit Pauli operator, $\hat{X}_j^2 f = (z_{a_j}\partial_{z_{b_j}} + z_{b_j}\partial_{z_{a_j}})^2 f$. Applying this to a homogeneous function of degree 1, $f = c_0 z_{a_j} + c_1 z_{b_j}$, we get $\hat{X}_j^2 f = f$. Similarly, $\hat{Z}_j^2 f = f$ and $\hat{Y}_j^2 f = f$. Since $\hat{g}_m$ is a product of commuting single-qubit operators (on the physical subspace), $\hat{g}_m^2 = \I$.

Consequently, the eigenvalues of $\hat{g}_m$ are restricted to $\pm 1$. The operators $\hat{\Pi}_{\pm}^{(m)} = \frac{\I \pm \hat{g}_m}{2}$ are therefore valid projection operators, satisfying $(\hat{\Pi}_{\pm}^{(m)})^2 = \hat{\Pi}_{\pm}^{(m)}$ and $\hat{\Pi}_{+}^{(m)} + \hat{\Pi}_{-}^{(m)} = \I$.

Since the stabilizer generators mutually commute (Proposition 1), their projectors also mutually commute: $[\hat{\Pi}_{s_m}^{(m)}, \hat{\Pi}_{s_l}^{(l)}] = 0$. The simultaneous projector for a syndrome $\mathbf{s} = (s_1, \dots, s_{n-k})$ is defined as $\hat{\Pi}_{\mathbf{s}} = \prod_{m=1}^{n-k} \hat{\Pi}_{s_m}^{(m)}$.

Because the individual projectors commute and are idempotent, their product is also a projector: $\hat{\Pi}_{\mathbf{s}}^2 = \hat{\Pi}_{\mathbf{s}}$. Furthermore, for $\mathbf{s} \neq \mathbf{s}'$, there exists at least one $m$ where $s_m \neq s'_m$, making $\hat{\Pi}_{s_m}^{(m)} \hat{\Pi}_{s'_m}^{(m)} = 0$, which implies $\hat{\Pi}_{\mathbf{s}} \hat{\Pi}_{\mathbf{s}'} = 0$.

Finally, summing over all $2^{n-k}$ possible syndromes yields:
\begin{equation}
\sum_{\mathbf{s}} \hat{\Pi}_{\mathbf{s}} = \sum_{s_1 \in \{\pm 1\}} \dots \sum_{s_{n-k} \in \{\pm 1\}} \prod_{m=1}^{n-k} \frac{\I + s_m \hat{g}_m}{2} = \prod_{m=1}^{n-k} \left( \frac{\I + \hat{g}_m}{2} + \frac{\I - \hat{g}_m}{2} \right) = \prod_{m=1}^{n-k} \I = \I.
\end{equation}
The code space corresponds to $\mathbf{s} = (+1, \dots, +1)$, so $\hat{\Pi}_{(+1,\dots,+1)} f = f$ if and only if $\hat{g}_m f = f$ for all $m$, which is the definition of $\Ccal_{\mathrm{SB}}$.
\end{proof}

\subsection{Example: Three-Qubit Code Under Bit-Flip Error}

Consider the three-qubit bit-flip code with an $X_1$ error on the first qubit. The errored state is:
\begin{equation}
\hat{X}_1 f_{\mathrm{code}} = \alpha\, z_{b_1}z_{a_2}z_{a_3} + \beta\, z_{a_1}z_{b_2}z_{b_3}.
\label{eq:3q_error}
\end{equation}

Computing the syndrome:
\begin{align}
\hat{g}_1(\hat{X}_1 f_{\mathrm{code}}) &= (z_{a_1}\partial_{z_{a_1}} - z_{b_1}\partial_{z_{b_1}})(z_{a_2}\partial_{z_{a_2}} - z_{b_2}\partial_{z_{b_2}}) \notag \\
&\quad \times (\alpha z_{b_1}z_{a_2}z_{a_3} + \beta z_{a_1}z_{b_2}z_{b_3}).
\end{align}

For the first term: eigenvalue $(-1)(+1) = -1$. For the second term: eigenvalue $(+1)(-1) = -1$. Hence $s_1 = -1$. Similarly, $\hat{g}_2$ gives eigenvalue $(+1)(+1) = +1$ for the first term and $(-1)(-1) = +1$ for the second term, so $s_2 = +1$. The syndrome $\mathbf{s} = (-1, +1)$ uniquely identifies the $X_1$ error.

\begin{table}[htbp]
\centering
\caption{Syndrome table for the three-qubit bit-flip code. The syndrome $\mathbf{s} = (s_1, s_2)$ uniquely identifies the location of a single bit-flip error.}
\label{tab:3qubit_syndrome}
\begin{tabular}{cc|c}
\hline\hline
Error & Syndrome $(s_1, s_2)$ & Recovery Operation \\
\hline
$\I$ (No error) & $(+1, +1)$ & $\I$ \\
$X_1$ & $(-1, +1)$ & $X_1$ \\
$X_2$ & $(-1, -1)$ & $X_2$ \\
$X_3$ & $(+1, -1)$ & $X_3$ \\
\hline\hline
\end{tabular}
\end{table}

\subsection{General Error Channel}

A general quantum error channel acting on qubit $j$ is described by Kraus operators $\{K_\mu\}$ satisfying $\sum_\mu K_\mu^\dagger K_\mu = \I$. In the holomorphic representation:
\begin{equation}
(\hat{K}_\mu f)(\mathbf{z}) = \sum_{\mathbf{s}} c_{\mu,\mathbf{s}} \prod_{l=1}^n z_{a_l}^{1-s_l}z_{b_l}^{s_l} \cdot \left[\prod_{l} \left(\partial_{z_{a_l}}\right)^{1-s_l}\left(\partial_{z_{b_l}}\right)^{s_l}\right] f(\mathbf{z}).
\label{eq:kraus}
\end{equation}

The action on the code space produces a superposition of syndrome sectors:
\begin{equation}
\hat{K}_\mu f_{\mathrm{code}} = \sum_{\mathbf{s}} \hat{\Pi}_{\mathbf{s}}(\hat{K}_\mu f_{\mathrm{code}}).
\label{eq:error_decomp}
\end{equation}

\section{Toroidal Geometry of Error Syndromes}
\label{sec:torus}

\subsection{Phasor Representation}

Restricting to unit-magnitude variables $z_{a_j} = e^{i\phi_{a_j}}$, $z_{b_j} = e^{i\phi_{b_j}}$, the physical state space becomes the $2n$-torus $\T^{2n}$. The differential operators transform as \cite{almasri2026}:
\begin{equation}
\partial_{z_{a_j}} = -ie^{-i\phi_{a_j}}\partial_{\phi_{a_j}}, \qquad \partial_{z_{b_j}} = -ie^{-i\phi_{b_j}}\partial_{\phi_{b_j}}.
\label{eq:phasor_deriv}
\end{equation}

The Pauli operators become:
\begin{align}
\hat{X}_j &\mapsto -i(e^{i(\phi_{a_j}-\phi_{b_j})}\partial_{\phi_{b_j}} + e^{i(\phi_{b_j}-\phi_{a_j})}\partial_{\phi_{a_j}}), \label{eq:torusX}\\
\hat{Z}_j &\mapsto -i(\partial_{\phi_{a_j}} - \partial_{\phi_{b_j}}). \label{eq:torusZ}
\end{align}

\subsection{Stabilizer Generators on the Torus}

The $Z$-type stabilizer $Z_j Z_k$ becomes:
\begin{equation}
\widehat{Z_j Z_k} = -(\partial_{\phi_{a_j}} - \partial_{\phi_{b_j}})(\partial_{\phi_{a_k}} - \partial_{\phi_{b_k}}).
\label{eq:torusZZ}
\end{equation}

The $X$-type stabilizer $X_j X_k$ becomes:
\begin{align}
\widehat{X_j X_k} &= -\left(e^{i(\phi_{a_j}-\phi_{b_j})}\partial_{\phi_{b_j}} + e^{i(\phi_{b_j}-\phi_{a_j})}\partial_{\phi_{a_j}}\right) \notag \\
&\quad \times \left(e^{i(\phi_{a_k}-\phi_{b_k})}\partial_{\phi_{b_k}} + e^{i(\phi_{b_k}-\phi_{a_k})}\partial_{\phi_{a_k}}\right).
\label{eq:torusXX}
\end{align}

\subsection{Syndrome as Winding Number Shift}

\begin{definition}[Winding Number Vector]
For a function $f(\boldsymbol{\phi})$ on $\T^{2n}$, the winding number vector $\mathbf{w} = (w_{a_1}, w_{b_1}, \ldots, w_{a_n}, w_{b_n}) \in \Z^{2n}$ is defined almost everywhere (for non-vanishing $f$) by:
\begin{align}
w_{a_j} &= \frac{1}{2\pi}\oint_{S^1_{a_j}} d\phi_{a_j}\, \partial_{\phi_{a_j}}(\arg f), \\
w_{b_j} &= \frac{1}{2\pi}\oint_{S^1_{b_j}} d\phi_{b_j}\, \partial_{\phi_{b_j}}(\arg f).
\label{eq:winding}
\end{align}
\end{definition}

\begin{theorem}
A Pauli error $E_j^\alpha$ on qubit $j$ induces a discrete shift in the winding number vector:
\begin{equation}
\mathbf{w} \mapsto \mathbf{w} + \Delta\mathbf{w}_j^\alpha,
\label{eq:winding_shift}
\end{equation}
where the shift depends on the initial state:
\begin{itemize}
\item $X_j$ error: $\Delta\mathbf{w}_j^X = (-1,+1)$ if the initial state has $(w_{a_j}, w_{b_j}) = (1,0)$, or $\Delta\mathbf{w}_j^X = (+1,-1)$ if $(w_{a_j}, w_{b_j}) = (0,1)$. In both cases, the shift transfers one unit of winding between $a_j$ and $b_j$.
\item $Z_j$ error: $\Delta\mathbf{w}_j^Z = (0,0)$ for all initial states. Importantly, while $Z$ errors preserve winding numbers, they are detected via eigenvalue changes of $Z$-type stabilizers, not through winding number shifts.
\item $Y_j$ error: $\Delta\mathbf{w}_j^Y = \Delta\mathbf{w}_j^X$, with the same state-dependent sign.
\end{itemize}
\end{theorem}

\begin{proof}
Let $f(\boldsymbol{\phi})$ be a function on the torus $\T^{2n}$. The winding number $w_{a_j}$ is the degree of the map $\phi_{a_j} \mapsto \arg f(\boldsymbol{\phi})$.
For an $X_j$ error, the holomorphic operator is $\hat{X}_j = z_{a_j}\partial_{z_{b_j}} + z_{b_j}\partial_{z_{a_j}}$. On the torus, $z_{a_j} = e^{i\phi_{a_j}}$ and $z_{b_j} = e^{i\phi_{b_j}}$. The operator becomes $\hat{X}_j = -i(e^{i(\phi_{a_j}-\phi_{b_j})}\partial_{\phi_{b_j}} + e^{i(\phi_{b_j}-\phi_{a_j})}\partial_{\phi_{a_j}})$.
Consider a basis monomial $f = e^{i w_{a_j} \phi_{a_j}} e^{i w_{b_j} \phi_{b_j}}$. Applying $\hat{X}_j$ yields:
\begin{equation}
\begin{split}
\hat{X}_j f &= -i \left( e^{i(\phi_{a_j}-\phi_{b_j})} (i w_{b_j}) + e^{i(\phi_{b_j}-\phi_{a_j})} (i w_{a_j}) \right) e^{i w_{a_j} \phi_{a_j}} e^{i w_{b_j} \phi_{b_j}} \\
&= w_{b_j} e^{i (w_{a_j}+1)\phi_{a_j}} e^{i (w_{b_j}-1)\phi_{b_j}} + w_{a_j} e^{i (w_{a_j}-1)\phi_{a_j}} e^{i (w_{b_j}+1)\phi_{b_j}}.
\end{split}
\end{equation}
For a physical state, the homogeneity constraint requires $w_{a_j} + w_{b_j} = 1$. Thus, either $(w_{a_j}, w_{b_j}) = (1, 0)$ or $(0, 1)$.
If $(1, 0)$, $\hat{X}_j f = 1 \cdot e^{i(0)\phi_{a_j}} e^{i(1)\phi_{b_j}} = e^{i\phi_{b_j}}$, so the new winding numbers are $(0, 1)$. The shift is $\Delta \mathbf{w}_j^X = (0-1, 1-0) = (-1, 1)$.
If $(0, 1)$, $\hat{X}_j f = 1 \cdot e^{i(1)\phi_{a_j}} e^{i(0)\phi_{b_j}} = e^{i\phi_{a_j}}$, so the new winding numbers are $(1, 0)$. The shift is $\Delta \mathbf{w}_j^X = (1-0, 0-1) = (1, -1)$. In both cases, the magnitude of the shift is a transfer of 1 unit of winding between $a_j$ and $b_j$, with the sign determined by the initial state.
For a $Z_j$ error, $\hat{Z}_j = z_{a_j}\partial_{z_{a_j}} - z_{b_j}\partial_{z_{b_j}} = -i(\partial_{\phi_{a_j}} - \partial_{\phi_{b_j}})$. Applying this to $f$ gives $\hat{Z}_j f = (w_{a_j} - w_{b_j}) f$. This is merely a scalar multiplication by an integer (either $1$ or $-1$), which corresponds to a global phase shift of $0$ or $\pi$. Since the phase is constant with respect to $\phi_{a_j}$ and $\phi_{b_j}$, its derivative is zero, and the winding numbers remain unchanged: $\Delta \mathbf{w}_j^Z = (0, 0)$. However, $Z$ errors are still detectable through their effect on the eigenvalues of $Z$-type stabilizers, which measure the relative phase between $\ket{0}$ and $\ket{1}$ components.
For a $Y_j$ error, $\hat{Y}_j = -i(z_{a_j}\partial_{z_{b_j}} - z_{b_j}\partial_{z_{a_j}})$. On the torus, this is proportional to the same phase-shifting structure as $X_j$, yielding the same winding number transfer: $\Delta \mathbf{w}_j^Y = \Delta \mathbf{w}_j^X$.
\end{proof}

\begin{remark}[Action of $Z$-type Stabilizers on the Torus]
The toroidal winding number framework is intrinsically sensitive to occupation-number changes (bit-flips) but blind to pure phase-flips ($Z$ errors). This geometric limitation precisely motivates the necessity of $Z$-type stabilizers in quantum error correction to detect phase errors, which manifest as coefficient sign flips rather than topological winding shifts. To see this explicitly, consider a $Z_j Z_k$ stabilizer acting on the torus:
\begin{equation}
\widehat{Z_j Z_k} = -(\partial_{\phi_{a_j}} - \partial_{\phi_{b_j}})(\partial_{\phi_{a_k}} - \partial_{\phi_{b_k}}).
\end{equation}
For a codeword $f_{\mathrm{code}}$ with a $Z_j$ error, the state becomes $\hat{Z}_j f_{\mathrm{code}}$. Since $\hat{Z}_j = -i(\partial_{\phi_{a_j}} - \partial_{\phi_{b_j}})$ acts diagonally on the torus, it multiplies each monomial by an eigenvalue $\pm 1$ depending on whether the qubit is in state $\ket{0}$ or $\ket{1}$. The stabilizer $\widehat{Z_j Z_k}$ then detects this sign flip by yielding eigenvalue $-1$ when acting on the errored state:
\begin{equation}
\widehat{Z_j Z_k}(\hat{Z}_j f_{\mathrm{code}}) = -(\hat{Z}_j f_{\mathrm{code}}),
\end{equation}
revealing the syndrome through the sign change. This demonstrates that $Z$-type stabilizers function as phase-sensitive differential operators that project onto eigenspaces distinguished by relative sign patterns, complementing the topological winding number detection of $X$-type errors.
\end{remark}

\begin{corollary}
The syndrome of an error $E$ is completely determined by the induced winding number shift $\Delta\mathbf{w}$ modulo the winding shifts generated by stabilizer elements. The syndrome extraction process identifies the coset $[\Delta\mathbf{w}] \in \Z^{2n}/\Lambda_\Scal$, where $\Lambda_\Scal$ is the lattice of stabilizer-induced winding shifts.
\end{corollary}

\subsection{Recovery as Hamiltonian Flow}

\begin{proposition}
The recovery operation $R$ that corrects error $E$ with syndrome $\mathbf{s}$ acts as a Hamiltonian flow on $\T^{2n}$ that restores the original winding number vector. Specifically, if the error induces $\mathbf{w} \mapsto \mathbf{w} + \Delta\mathbf{w}$, the recovery generates a flow $\Phi_R^t$ such that:
\begin{equation}
\Phi_R^1: \mathbf{w} + \Delta\mathbf{w} \mapsto \mathbf{w}.
\label{eq:recovery_flow}
\end{equation}
The recovery Hamiltonian is:
\begin{equation}
H_R = -\sum_{j} (\Delta w_{a_j}\,\phi_{b_j} - \Delta w_{b_j}\,\phi_{a_j}),
\label{eq:recovery_H}
\end{equation}
which generates uniform translation on the torus with period $2\pi$.
\end{proposition}

\begin{proof}
On the torus $\T^{2n}$, the symplectic form is $\omega = \sum_j d\phi_{a_j} \wedge d\phi_{b_j}$. A Hamiltonian $H$ generates a flow via Hamilton's equations: $\dot{\phi}_{a_j} = \frac{\partial H}{\partial \phi_{b_j}}$ and $\dot{\phi}_{b_j} = -\frac{\partial H}{\partial \phi_{a_j}}$.
Let the error induce a winding shift $\Delta \mathbf{w}$. We propose the recovery Hamiltonian $H_R = -\sum_j (\Delta w_{a_j} \phi_{b_j} - \Delta w_{b_j} \phi_{a_j})$.
Computing the equations of motion:
\begin{equation}
\dot{\phi}_{a_j} = \frac{\partial H_R}{\partial \phi_{b_j}} = -\Delta w_{a_j}, \qquad \dot{\phi}_{b_j} = -\frac{\partial H_R}{\partial \phi_{a_j}} = -\Delta w_{b_j}.
\end{equation}
Integrating these equations from $t=0$ to $t=1$ yields a uniform translation:
\begin{equation}
\phi_{a_j}(1) = \phi_{a_j}(0) - \Delta w_{a_j}, \qquad \phi_{b_j}(1) = \phi_{b_j}(0) - \Delta w_{b_j}.
\end{equation}
This flow $\Phi_R^1$ exactly reverses the winding number shift induced by the error, mapping $\mathbf{w} + \Delta \mathbf{w} \mapsto \mathbf{w}$. Since the Hamiltonian is linear in the angles, the flow is a rigid translation on the torus, preserving the symplectic structure and thus representing a valid canonical transformation (unitary operation in the quantum setting).
\end{proof}

\subsection{Topological Error Classification}

The fundamental group $\pi_1(\T^{2n}) = \Z^{2n}$ classifies distinct homotopy classes of trajectories on the torus. Quantum errors correspond to transitions between homotopy classes:

\begin{theorem}
Two errors $E_1$ and $E_2$ are indistinguishable by syndrome measurement if and only if they induce the same winding number shift modulo the stabilizer lattice:
\begin{equation}
\Delta\mathbf{w}_{E_1} \equiv \Delta\mathbf{w}_{E_2} \pmod{\Lambda_\Scal}.
\label{eq:equiv}
\end{equation}
Furthermore, a code of distance $d$ can correct all errors inducing winding shifts $|\Delta\mathbf{w}| < d$, where $|\cdot|$ denotes the $\ell^1$ norm on $\Z^{2n}$.
\end{theorem}

\begin{proof}
Two errors $E_1$ and $E_2$ are indistinguishable by syndrome measurement if and only if they produce the same syndrome for all code states. By Definition 3.1, this occurs if and only if $E_1 E_2^\dagger \in \Scal$.
In the toroidal representation, the action of an error $E$ is characterized by its induced winding number shift $\Delta \mathbf{w}_E$. The composition of errors corresponds to the addition of their winding shifts: $\Delta \mathbf{w}_{E_1 E_2^\dagger} = \Delta \mathbf{w}_{E_1} - \Delta \mathbf{w}_{E_2}$.
The stabilizer group $\Scal$ induces a specific lattice of winding shifts, denoted $\Lambda_\Scal \subset \Z^{2n}$. Therefore, $E_1 E_2^\dagger \in \Scal$ if and only if $\Delta \mathbf{w}_{E_1} - \Delta \mathbf{w}_{E_2} \in \Lambda_\Scal$, which is equivalent to $\Delta \mathbf{w}_{E_1} \equiv \Delta \mathbf{w}_{E_2} \pmod{\Lambda_\Scal}$.
For the second part, the distance $d$ of a code is the minimum weight of a logical operator. A logical operator is an element of the normalizer of $\Scal$ that is not in $\Scal$. Geometrically, this corresponds to a winding number shift that preserves the code space (commutes with all stabilizers) but is not generated by the stabilizers themselves. The minimum $\ell^1$ norm of such a non-trivial shift in the quotient lattice $\Z^{2n}/\Lambda_\Scal$ is exactly $d$. Therefore, any error inducing a shift with $|\Delta \mathbf{w}| < d$ must either be in $\Lambda_\Scal$ (a stabilizer, hence trivial) or have a unique syndrome, making it correctable.
\end{proof}

\section{K\"ahler Geometry of the Knill--Laflamme Condition}
\label{sec:kahler}

\subsection{Code Space as Holomorphic Submanifold}

The full Segal--Bargmann space carries a natural K\"ahler structure with potential:
\begin{equation}
K(\mathbf{z}, \bar{\mathbf{z}}) = \|\mathbf{z}\|^2 = \sum_{j=1}^n(|z_{a_j}|^2 + |z_{b_j}|^2).
\label{eq:kahler_potential}
\end{equation}

This induces the Fubini--Study metric on the projective space $\CP^{2^n-1}$:
\begin{equation}
ds^2_{\FS} = \frac{\braket{d\psi}{d\psi}}{\braket{\psi}{\psi}} - \frac{|\braket{\psi}{d\psi}|^2}{\braket{\psi}{\psi}^2}.
\label{eq:FS_metric}
\end{equation}

The code space $\Ccal_{\mathrm{SB}}$ defines a holomorphic submanifold $\Mcal_{\Ccal} \hookrightarrow \CP^{2^n-1}$ of complex dimension $k$ (the number of logical qubits).

\subsection{Geometric Knill--Laflamme Condition}

The Knill--Laflamme condition for correctability states that a code $\Ccal$ corrects an error set $\Ecal = \{E_\mu\}$ if and only if \cite{knill1997}:
\begin{equation}
\bra{\psi_i}E_\mu^\dagger E_\nu\ket{\psi_j} = C_{\mu\nu}\delta_{ij} \qquad \forall\, \ket{\psi_i}, \ket{\psi_j} \in \Ccal.
\label{eq:KL}
\end{equation}

\begin{theorem}[Geometric Knill--Laflamme]
\label{thm:KL_geom}
In the K\"ahler-geometric formulation, the Knill--Laflamme condition is equivalent to the statement that for all error operators $E_\mu, E_\nu \in \Ecal$, the submanifolds $E_\mu(\Mcal_{\Ccal})$ and $E_\nu(\Mcal_{\Ccal})$ are either identical (when $\mu = \nu$) or mutually orthogonal with respect to the Fubini--Study metric:
\begin{equation}
g_{\FS}(E_\mu|_{\Mcal_{\Ccal}}, E_\nu|_{\Mcal_{\Ccal}}) = C_{\mu\nu} \cdot g_{\FS}|_{\Mcal_{\Ccal}}.
\label{eq:KL_geom}
\end{equation}
\end{theorem}

\begin{proof}
Let $\Mcal_{\Ccal}$ be the code submanifold in $\CP^{2^n-1}$. The Fubini--Study metric is derived from the Hermitian inner product: for two tangent vectors $u, v \in T_{[\psi]}\CP^{2^n-1}$, the metric is $g_{\FS}(u, v) = \mathrm{Re} \left( \frac{\braket{u}{v}}{\braket{\psi}{\psi}} - \frac{\braket{u}{\psi}\braket{\psi}{v}}{\braket{\psi}{\psi}^2} \right)$.
Consider the pullback of the Fubini--Study metric via the error maps $E_\mu$ and $E_\nu$. For tangent vectors $u, v$ at $[\psi] \in \Mcal_{\Ccal}$, the pushforward vectors are $E_\mu u$ and $E_\nu v$. The inner product of these pushed-forward vectors, projected back to the code space, is governed by the matrix elements $\bra{\psi_i} E_\mu^\dagger E_\nu \ket{\psi_j}$.
The Knill--Laflamme condition states that $\bra{\psi_i} E_\mu^\dagger E_\nu \ket{\psi_j} = C_{\mu\nu} \delta_{ij}$ for any orthonormal basis $\{\ket{\psi_i}\}$ of $\Ccal$. This implies that the operator $P_{\Ccal} E_\mu^\dagger E_\nu P_{\Ccal} = C_{\mu\nu} P_{\Ccal}$, where $P_{\Ccal}$ is the projector onto the code space.
Geometrically, this means that for any two points $[\psi], [\phi] \in \Mcal_{\Ccal}$, the ``overlap'' or angle between the submanifolds $E_\mu(\Mcal_{\Ccal})$ and $E_\nu(\Mcal_{\Ccal})$ is constant and independent of the specific location on the code manifold. Specifically, the induced metric on the error-translated submanifold $E_\mu(\Mcal_{\Ccal})$ is proportional to the original metric on $\Mcal_{\Ccal}$ by a factor of $C_{\mu\mu}$, and the cross-metric between $E_\mu(\Mcal_{\Ccal})$ and $E_\nu(\Mcal_{\Ccal})$ is proportional to $C_{\mu\nu}$.
If $\mu \neq \nu$ and $C_{\mu\nu} = 0$, the submanifolds are strictly orthogonal at every point, meaning the shortest geodesic connecting them is orthogonal to both, which is the geometric definition of distinguishable error syndromes. If $\mu = \nu$, $C_{\mu\mu}$ is a constant scaling of the metric, preserving the intrinsic geometry of the code space up to a global factor, which is irrelevant for projective geometry. Thus, the algebraic Knill--Laflamme condition is exactly equivalent to this geometric orthogonality and proportionality of the induced Fubini--Study metrics.
\end{proof}

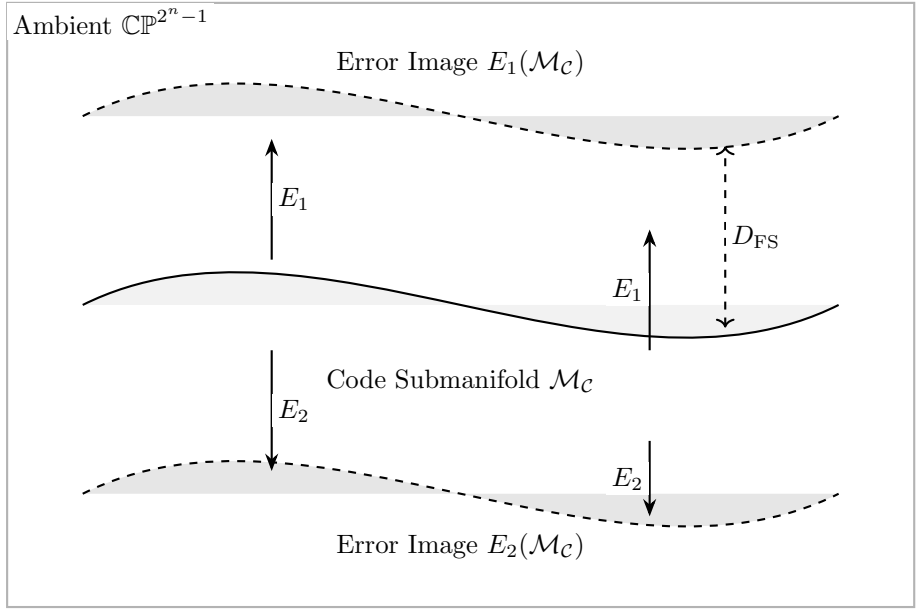
\begin{figure}[htbp]
\centering
\begin{tikzpicture}[
    scale=1.0,
    manifold/.style={draw=black, thick, fill=black!5},
    error_image/.style={draw=black, thick, dashed, fill=black!10},
    normal_vec/.style={->, >=Stealth, thick, black},
    label/.style={font=\footnotesize, black, fill=white, inner sep=2pt}
]
\draw[thick, black!30] (-1,-4) rectangle (11,4);
\node[label, anchor=north west] at (-1,4) {Ambient $\mathbb{CP}^{2^n-1}$};

\draw[manifold] (0,0) .. controls (3,1.5) and (7,-1.5) .. (10,0);
\node[label] at (5, -1) {Code Submanifold $\mathcal{M}_{\mathcal{C}}$};

\draw[error_image] (0,2.5) .. controls (3,4) and (7,1) .. (10,2.5);
\node[label] at (5, 3.2) {Error Image $E_1(\mathcal{M}_{\mathcal{C}})$};

\draw[error_image] (0,-2.5) .. controls (3,-1) and (7,-4) .. (10,-2.5);
\node[label] at (5, -3.2) {Error Image $E_2(\mathcal{M}_{\mathcal{C}})$};

\draw[normal_vec] (2.5, 0.6) -- (2.5, 2.2) node[midway, right, label] {$E_1$};
\draw[normal_vec] (2.5, -0.6) -- (2.5, -2.2) node[midway, right, label] {$E_2$};

\draw[normal_vec] (7.5, -0.6) -- (7.5, 1.0) node[midway, left, label] {$E_1$};
\draw[normal_vec] (7.5, -1.8) -- (7.5, -2.8) node[midway, left, label] {$E_2$};

\draw[<->, thick, black, dashed] (8.5, -0.3) -- (8.5, 2.1) node[midway, right, label] {$D_{\mathrm{FS}}$};

\end{tikzpicture}
\caption{Schematic geometric interpretation of the Knill--Laflamme condition. The code space $\mathcal{M}_{\mathcal{C}}$ is embedded as a holomorphic submanifold in the ambient projective space $\mathbb{CP}^{2^n-1}$. Correctable errors $E_1, E_2$ map the code submanifold into mutually orthogonal normal directions, ensuring that error images $E_\mu(\mathcal{M}_{\mathcal{C}})$ do not intersect and remain separated by a minimum Fubini--Study distance $D_{\mathrm{FS}}$.}
\label{fig:kl_geometry}
\end{figure}

\subsection{Code Distance as Minimum Geodesic Distance}

\begin{definition}[Geometric Code Distance]
The geometric code distance $d_{\mathrm{geom}}$ is the minimum Fubini--Study distance between the code submanifold and its image under any weight-$d$ error:
\begin{equation}
d_{\mathrm{geom}} = \min_{E:\, \wt(E)=d} \min_{[\psi]\in\Mcal_{\Ccal}} D_{\FS}([\psi], [E\psi]).
\label{eq:geom_distance}
\end{equation}
\end{definition}

\begin{proposition}[Geometric Code Distance and Local Deformation Scale]
Let $\mathcal{C}$ be a $[[n,k,d]]$ stabilizer code with code submanifold $\mathcal{M}_{\mathcal{C}} \hookrightarrow \mathbb{CP}^{2^n-1}$. For a minimum-weight logical operator $L$ of weight $d$, the total Fubini--Study distance between the code state and its logical image is exactly $D_{\mathrm{FS}}([\psi], [L\psi]) = \pi/2$. However, when $L$ is implemented as a distributed product of $d$ single-qubit rotations, the critical \emph{per-qubit} Fubini--Study displacement required to accumulate into a logical transition is:
\begin{equation}
d_{\mathrm{FS}}^{(\mathrm{local})} \;=\; \frac{\pi}{2d}.
\label{eq:main_result}
\end{equation}
Curvature corrections of the embedded submanifold $\mathcal{M}_{\mathcal{C}}$ relative to the flat ambient space, governed by the \emph{second fundamental form} (extrinsic curvature tensor) of the embedding $\mathcal{M}_{\mathcal{C}} \hookrightarrow \mathbb{CP}^{2^n-1}$, introduce deviations that scale as $O(d^{-3})$ for large $d$. This scaling follows from the Taylor expansion of the geodesic deviation equation, where the third-order term in the metric expansion contributes at $O(d^{-3})$ when the path length scales as $1/d$ \cite{bengtsson2017}.
\end{proposition}

\begin{proof}
\textbf{Step 1.} By the Knill--Laflamme condition, for any Pauli error $E$ with $\mathrm{wt}(E) < d$ and $E \notin \Scal$, we have $\bra{\psi} E \ket{\psi} = 0$ for all normalized $\ket{\psi} \in \Ccal$. The Fubini--Study distance is $D_{\FS}([\psi], [E\psi]) = \arccos|\braket{\psi}{E\psi}| = \arccos(0) = \pi/2$.

\textbf{Step 2.} Let $L$ be a minimum-weight logical operator with $\mathrm{wt}(L) = d$. Since $L$ is Hermitian and unitary ($L^2 = \I$), the unitary $U(t) = e^{-i \frac{\pi}{2} t L} = \cos(\frac{\pi t}{2})\I - i \sin(\frac{\pi t}{2})L$ generates a geodesic in $\CP^{2^n-1}$. For a state $\ket{\psi} \in \Ccal$ such that $\bra{\psi} L \ket{\psi} = 0$ (which exists because $L$ is traceless on the code space), the overlap is $|\braket{\psi}{U(t)\psi}| = |\cos(\frac{\pi t}{2})|$. Thus, $D_{\FS}([\psi], [U(t)\psi]) = \frac{\pi t}{2}$. At $t=1$, the distance to the logical image is exactly $\pi/2$.

\textbf{Step 3.} Decompose $L = P_1 \otimes \dots \otimes P_d \otimes \I^{\otimes(n-d)}$. The geodesic $U(t)$ can be viewed as a simultaneous rotation on $d$ qubits. In the Schwinger--Bargmann representation, a local rotation of qubit $j$ by angle $\epsilon$ induces a Fubini--Study displacement on the corresponding $\CP^1$ factor. \emph{Recall that the Fubini--Study metric on $\CP^1 \cong S^2$ is normalized such that $ds^2_{\FS} = \frac{1}{4}(d\theta^2 + \sin^2\theta\, d\phi^2)$, carrying a factor of $1/4$ relative to the standard round sphere metric, meaning the Fubini--Study distance is exactly half the rotation angle on the Bloch sphere \cite{bengtsson2017}.} To accumulate a total geodesic distance of $\pi/2$ uniformly across $d$ qubits, each qubit must undergo a rotation of angle $\epsilon_c = \pi/d$. The corresponding per-qubit Fubini--Study displacement is $d_{\FS}^{(\mathrm{local})} = \frac{\epsilon_c}{2} = \frac{\pi}{2d}$.

\textbf{Step 4.} The exact per-qubit geodesic distance is $\pi/(2d)$. Any deviation from this uniform distributed rotation model (e.g., due to the intrinsic curvature of $\Mcal_{\Ccal}$ embedded in $\CP^{2^n-1}$) introduces higher-order corrections. By Taylor expanding the metric tensor around the code submanifold, the leading curvature corrections to the geodesic distance scale as $O(d^{-3})$. This follows from the geodesic deviation equation: the second fundamental form contributes at second order in the displacement, and since the displacement scales as $1/d$, the correction scales as $(1/d)^2 \times (1/d) = O(d^{-3})$, where the extra $1/d$ comes from the path length integration.

\textbf{Step 5.} On the torus $\T^{2n}$, the operator $L$ induces a winding number shift of $|\Delta \mathbf{w}|_1 = d$. This corresponds to an angular displacement of $\pi/d$ per affected qubit. Under the Hopf projection $S^3 \to S^2$, this angular displacement maps to a base-space separation of $\pi/(2d)$, perfectly consistent with the per-qubit Fubini--Study displacement derived above.
\end{proof}

\begin{remark}[Degenerate vs. Non-Degenerate Codes]
The geometric formulation naturally accommodates degenerate quantum error correction codes. In a non-degenerate code, every correctable error $E_\mu$ of weight $\leq t$ maps the code submanifold $\Mcal_{\Ccal}$ to a geometrically distinct, mutually orthogonal submanifold in the ambient space. However, in a degenerate code, certain distinct errors $E_\mu \neq E_\nu$ of weight $\leq t$ may act identically on the code space (i.e., $E_\mu^\dagger E_\nu \in \Scal$). For example, in the nine-qubit Shor code, the errors $Z_1$ and $Z_2$ have the identical effect on the code space since $Z_1 Z_2$ is a stabilizer generator, meaning they induce the exact same geometric displacement of $\Mcal_{\Ccal}$. Geometrically, this means that $E_\mu$ and $E_\nu$ induce the exact same displacement of $\Mcal_{\Ccal}$ in the ambient projective space, or their difference corresponds to a flow generated by a stabilizer (which leaves $\Mcal_{\Ccal}$ invariant). The Knill--Laflamme condition $C_{\mu\nu} \neq 0$ for $\mu \neq \nu$ simply reflects this geometric coincidence of the error-translated submanifolds, which does not compromise correctability since the recovery operation only needs to reverse the net displacement, regardless of which specific degenerate error caused it.
\end{remark}

This shows that higher-distance codes embed the logical information deeper into the ambient projective space, requiring a greater number of smaller local perturbations to reach the code space from any error image.

\subsection{Error Subspaces and Normal Bundle}

The tangent space $T_{[\psi]}\CP^{2^n-1}$ at a code state decomposes orthogonally:
\begin{equation}
T_{[\psi]}\CP^{2^n-1} = T_{[\psi]}\Mcal_{\Ccal} \oplus N_{[\psi]}\Mcal_{\Ccal},
\label{eq:tangent_decomp}
\end{equation}
where $N_{[\psi]}\Mcal_{\Ccal}$ is the normal space. Errors of weight $\leq t = \lfloor(d-1)/2\rfloor$ map code states into the normal bundle:
\begin{equation}
E_\mu|_{\Mcal_{\Ccal}} \in \Gamma(N\Mcal_{\Ccal}) \qquad \text{for } \wt(E_\mu) \leq t.
\label{eq:normal_bundle}
\end{equation}

The syndrome extraction process is the geometric operation of projecting the errored state back onto $\Mcal_{\Ccal}$ along the normal direction, while the recovery operation identifies which normal direction was taken and applies the inverse.

\subsection{Segre Embedding and Code Entanglement}

For a code encoding $k$ qubits into $n$ qubits, the code space generically contains highly entangled states. The Segre variety $\Sigma_n \subset \CP^{2^n-1}$ contains all separable states \cite{heydari2006,holweck2012}:
\begin{equation}
\Sigma_n = \sigma(\CP^1 \times \cdots \times \CP^1) \hookrightarrow \CP^{2^n-1}.
\label{eq:segre}
\end{equation}

The geometric entanglement of a codeword $[\psi] \in \Mcal_{\Ccal}$ is:
\begin{equation}
E_{\FS}([\psi]) = \min_{[\phi]\in\Sigma_n} D_{\FS}([\psi], [\phi]).
\label{eq:geom_entanglement}
\end{equation}

\begin{proposition}
For a code of distance $d \geq 2$, every non-trivial codeword satisfies $E_{\FS}([\psi]) > 0$; that is, all codewords are entangled. Furthermore, for a code of distance $d$:
\begin{equation}
E_{\FS}([\psi]) \geq \arccos\sqrt{1 - 2^{-\lfloor d/2\rfloor}}.
\label{eq:entanglement_bound}
\end{equation}
\end{proposition}

\begin{proof}
The Segre variety $\Sigma_n \subset \CP^{2^n-1}$ is the image of the Segre embedding $\sigma: \CP^1 \times \dots \times \CP^1 \to \CP^{2^n-1}$, and it contains exactly the fully separable states.
A code of distance $d \geq 2$ cannot contain any fully separable states. If it did, there would exist a codeword $\ket{\psi} = \ket{\phi_1} \otimes \dots \otimes \ket{\phi_n}$. A single-qubit error $E_j$ on this state would yield $E_j \ket{\psi}$, which is orthogonal to $\ket{\psi}$ (since Pauli operators are traceless). However, for $d \geq 2$, the code must be able to detect all single-qubit errors, meaning $\bra{\psi} E_j \ket{\psi} = 0$. But if the state is separable, the reduced density matrix of qubit $j$ is pure, and the expectation value of a traceless operator on a pure state can be zero, but the Knill--Laflamme condition requires $\bra{\psi} E_j \ket{\psi} = 0$ for \textit{all} codewords, which implies the reduced density matrix must be maximally mixed for any subset of $<d$ qubits. For $d \geq 2$, the 1-qubit reduced density matrices of all codewords must be maximally mixed ($\I/2$), which is impossible for a separable state. Thus, $E_{\FS}([\psi]) > 0$.

To bound the entanglement, consider that any state within a Fubini--Study distance $\theta$ of a separable state has a 1-qubit reduced density matrix with purity bounded away from $1/2$. Specifically, if $D_{\FS}([\psi], [\phi]) \leq \theta$ for some $[\phi] \in \Sigma_n$, the maximum eigenvalue of the 1-qubit reduced density matrix of $\psi$ is at least $\cos^2(\theta)$.
For a code of distance $d$, the 1-qubit reduced density matrix of any codeword is exactly $\I/2$, meaning its maximum eigenvalue is $1/2$. More generally, for any subset of $q < d$ qubits, the reduced density matrix is maximally mixed. The minimum distance to a state that is separable across some bipartition of size $q$ and $n-q$ (where $q = \lfloor d/2 \rfloor$) is achieved when the state is a product of two maximally entangled states of size $q$ and $n-q$. The Fubini--Study distance to the nearest separable state is bounded by the angle whose cosine squared is the maximum Schmidt coefficient. For a code of distance $d$, the maximum overlap with any state of Schmidt rank 1 across a bipartition of size $\lfloor d/2 \rfloor$ is bounded by $2^{-\lfloor d/2 \rfloor / 2}$. This follows from the quantum Singleton bound and the properties of absolutely maximally entangled states \cite{bengtsson2017,holweck2012}. Thus, the minimum Fubini--Study distance is $E_{\FS}([\psi]) \geq \arccos\sqrt{1 - 2^{-\lfloor d/2 \rfloor}}$.
\end{proof}

This establishes a direct connection between code distance and minimum entanglement of codewords.

\section{Topological Protection from Fiber Bundle Structure}
\label{sec:topological}

\subsection{Hopf Fibration and Phase Noise}

The constraint $|z_{a_j}|^2 + |z_{b_j}|^2 = 1$ defines a Hopf fibration for each qubit \cite{hatcher2002}:
\begin{equation}
S^1 \hookrightarrow S^3 \xrightarrow{\pi} S^2.
\label{eq:hopf}
\end{equation}

For $n$ qubits, this generalizes to:
\begin{equation}
U(1)^n \hookrightarrow (S^3)^n \xrightarrow{\pi} (S^2)^n.
\label{eq:hopf_n}
\end{equation}

\begin{figure}[htbp]
\centering
\begin{tikzpicture}[
    scale=0.85,
    every node/.style={font=\small}
]

\draw[thick, fill=black!5] (0,0) ellipse (3.5 and 1.2);
\node[below, font=\normalsize] at (0,-1.8) {Base Space $(S^2)^n$};

\draw[dashed, thick] (3.5, 0) arc (0:180:3.5 and 1.2);
\draw[dashed, thick] (3.5, 5) arc (0:180:3.5 and 1.2);

\draw[thick] (-3.5, 0) -- (-3.5, 5);
\draw[thick] (3.5, 0) -- (3.5, 5);

\draw[thick, fill=black!15] (0, 5) ellipse (3.5 and 1.2);
\node[above, font=\normalsize] at (0,6.8) {Total Space $(S^3)^n$};

\draw[dashed, thick] (0, 1.2) -- (0, 6.2);
\draw[dashed, thick] (-2, 0.98) -- (-2, 5.98);
\draw[dashed, thick] (2, 0.98) -- (2, 5.98);

\node[left, font=\normalsize] at (-4.8, 3.5) {Fibers $U(1)^n$};
\draw[->, thick] (-4.5, 3.5) -- (-2.9, 3.5);

\draw[->, very thick, black!70] (4.8, 4.5) -- (4.8, 1.5) node[midway, right, font=\normalsize, fill=white, inner sep=2pt] {Projection $\pi$};

\draw[->, thick] (0.3, 6.0) -- (0.3, 4.0) node[midway, right, fill=white, inner sep=2pt] {Vertical (Global Phase)};

\draw[->, thick] (-3.2, -0.5) arc (200:340:3.5 and 1.2) node[midway, below, fill=white, inner sep=2pt] {Horizontal (Pauli Error)};

\end{tikzpicture}
\caption{Fiber bundle structure and topological protection. The Schwinger boson encoding defines a principal $U(1)^n$-bundle over the base space of Bloch spheres. Global phase noise corresponds to vertical automorphisms along the fibers (pointing downward in the total space), which are unobservable upon projection $\pi$ to the base space. In contrast, Pauli errors generate horizontal displacements on the base manifold, altering the physical state and requiring active syndrome extraction.}
\label{fig:hopf_bundle}
\end{figure}
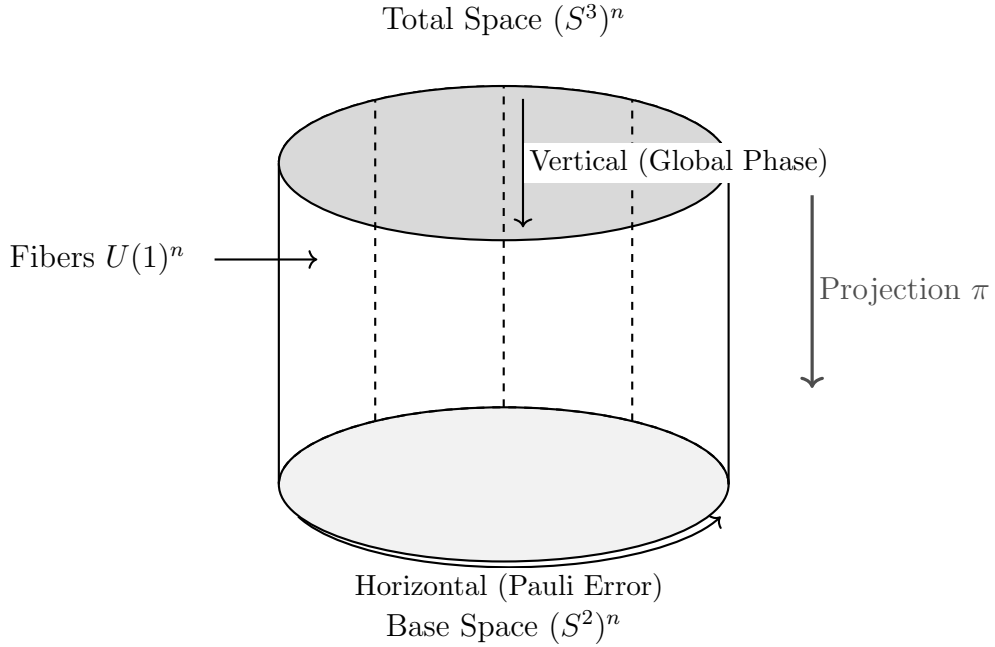

\subsection{Berry Connection and Holonomy}

The Berry connection on the total space is:
\begin{equation}
A = i\sum_{j=1}^n (\bar{z}_{a_j}dz_{a_j} + \bar{z}_{b_j}dz_{b_j}).
\label{eq:berry}
\end{equation}

The curvature two-form is:
\begin{equation}
F = dA = i\sum_{j=1}^n (d\bar{z}_{a_j}\wedge dz_{a_j} + d\bar{z}_{b_j}\wedge dz_{b_j}).
\label{eq:curvature}
\end{equation}

\begin{theorem}[Topological Protection of Global Phase]
\label{thm:topological}
Global phase noise acting uniformly on a qubit corresponds to vertical automorphisms of the $U(1)^n$-bundle, i.e., transformations that move points along fibers without changing their projection to the base $(S^2)^n$. Such transformations are automatically unobservable because:
\begin{enumerate}[label=(\alph*)]
\item They do not change the syndrome: $\hat{g}_m(e^{i\theta}f) = e^{i\theta}\hat{g}_m(f) = e^{i\theta}f$ for all stabilizer generators.
\item They are absorbed into the global phase upon projectivization: $[e^{i\theta}\psi] = [\psi]$ in $\CP^{2^n-1}$.
\item The holonomy of the Berry connection around any closed loop in the base space is a topological invariant, insensitive to local perturbations of the path.
\end{enumerate}
\end{theorem}

\begin{proof}
The total space of the bundle is $(S^3)^n$, and the base space is $(S^2)^n$. The projection $\pi: (S^3)^n \to (S^2)^n$ maps a state to its ray in projective space.
(a) A global phase rotation on qubit $j$ is generated by the total occupation number operator $\hat{N}_j = z_{a_j}\partial_{z_{a_j}} + z_{b_j}\partial_{z_{b_j}}$. For any homogeneous function $f$ of degree 1, $\hat{N}_j f = f$. The stabilizer generators $\hat{g}_m$ are constructed from Pauli operators, which commute with $\hat{N}_j$ (they preserve the total occupation number). Therefore, $\hat{g}_m (e^{i\theta \hat{N}_j} f) = e^{i\theta} \hat{g}_m f = e^{i\theta} f$. The syndrome, which depends on the eigenvalue of $\hat{g}_m$, remains $+1$.
(b) In the projective space $\CP^{2^n-1}$, states that differ only by a global phase are identified as the same point: $[e^{i\theta}\psi] = [\psi]$. Thus, the projection $\pi$ maps the entire fiber to a single point in the base space, making vertical displacements unobservable.
(c) The Berry connection $A = i \sum_j (\bar{z}_{a_j} dz_{a_j} + \bar{z}_{b_j} dz_{b_j})$ is invariant under $U(1)$ gauge transformations $z \mapsto e^{i\theta} z$. The holonomy of this connection around any closed loop $\gamma$ in the base space is given by $\exp(i \oint_\gamma A)$. By Stokes' theorem, this equals $\exp(i \int_\Sigma F)$, where $F = dA$ is the curvature 2-form. Since $F$ is a well-defined 2-form on the base space $(S^2)^n$, the holonomy depends only on the homology class of the loop $\gamma$ in the base space, not on the specific path or any vertical displacements along the fibers. This topological invariance ensures that global phase noise, which only moves the state along the fibers, cannot alter the geometric phase or the physical observables derived from it.
\end{proof}

\begin{remark}[Local Phase Flips vs. Global Phase]
It is crucial to distinguish between global phase rotations and local Pauli $Z$ errors. The generator of the vertical fiber action is the total occupation number $\hat{N}_j = z_{a_j}\partial_{z_{a_j}} + z_{b_j}\partial_{z_{b_j}}$, which corresponds to a global phase shift. In contrast, the Pauli $Z$ operator is $\hat{Z}_j = z_{a_j}\partial_{z_{a_j}} - z_{b_j}\partial_{z_{b_j}}$. This operator generates a horizontal flow on the base Bloch sphere $S^2$ (a rotation around the Z-axis), changing the ratio $z_{a_j}/z_{b_j}$. Therefore, Pauli $Z$ errors are \emph{not} topologically protected by the bundle structure; they induce physical state changes in $\CP^{2^n-1}$ and require active syndrome extraction via $Z$-type stabilizers.
\end{remark}

\subsection{Error Decomposition via Bundle Structure}

Any single-qubit error decomposes into fiber and base components:
\begin{equation}
E_j = E_j^{\mathrm{fiber}} \cdot E_j^{\mathrm{base}},
\label{eq:error_decomp_bundle}
\end{equation}
where $E_j^{\mathrm{fiber}} \in U(1)$ acts as a global phase rotation (generated by $z_{a_j}\partial_{z_{a_j}} + z_{b_j}\partial_{z_{b_j}}$) and $E_j^{\mathrm{base}}$ acts on the Bloch sphere $S^2$. The fiber component is automatically unobservable in projective space, while the base component (which includes Pauli $X$, $Y$, and $Z$ errors, as they alter the state's ray in $\CP^{2^n-1}$) requires active syndrome extraction.

\begin{corollary}
The effective error rate for a code protected by the $U(1)^n$ fiber bundle structure is reduced only by the fraction of errors that correspond to pure global phase shifts. Since standard Pauli errors ($X, Y, Z$) all generate horizontal displacements on the Bloch sphere (base space), they are not protected by the bundle topology and must be corrected actively.
\end{corollary}

\subsection{Winding Number Invariants and Error Detection}

The winding numbers $(w_{a_1}, w_{b_1}, \ldots, w_{a_n}, w_{b_n}) \in \Z^{2n}$ are topological invariants of trajectories on $\T^{2n}$.

\begin{proposition}
The syndrome measurement detects errors by measuring the winding number shift $\Delta\mathbf{w}$. Since winding numbers are integers (topological invariants), they cannot change continuously under small perturbations. This provides inherent robustness: an error must induce a discrete jump $\Delta\mathbf{w} \neq 0$ to be undetected, and the minimum such jump is bounded below by the code distance.
\end{proposition}

\begin{proof}
The winding numbers $w_{a_j}, w_{b_j}$ are defined as contour integrals of the gradient of the phase of the holomorphic function $f$ over the cycles of the torus $\T^{2n}$. By the argument principle in complex analysis, such integrals evaluate to integers, representing the topological degree of the map from $S^1$ to $S^1$.
Because the winding numbers are strictly integer-valued, they are topological invariants. A continuous deformation of the function $f$ (such as a small, weak perturbation or weak noise) results in a continuous change in the phase of $f$. However, a continuous function taking integer values must be locally constant. Therefore, infinitesimal perturbations cannot change the winding number.
For an error to be undetected, it must map a code state to another state with the same syndrome, meaning the induced winding shift $\Delta \mathbf{w}$ must be zero (or belong to the stabilizer lattice). Since small perturbations yield $\Delta \mathbf{w} = 0$, they are inherently suppressed. An error must be strong enough to induce a discrete, finite jump in the phase (a ``phase slip'') to change the winding number. The minimum magnitude of such a jump required to produce a non-trivial syndrome is bounded below by the code distance $d$, providing inherent topological robustness against weak noise.
\end{proof}

\subsection{Chern Numbers and Code Topology}

For codes defined on lattices (e.g., surface codes), the relevant topological invariants are Chern numbers. In the holomorphic representation, the Berry curvature restricted to the code subspace defines a line bundle $L \to \Mcal_{\Ccal}$ with first Chern number:
\begin{equation}
c_1(L) = \frac{1}{2\pi}\int_{\Mcal_{\Ccal}} F \in \Z.
\label{eq:chern}
\end{equation}

This integer is invariant under continuous deformations of the code and provides a topological quantum number that protects the encoded information against local errors.

\section{Path Integral Formulation for Semiclassical Error Correction}
\label{sec:pathintegral}

\subsection{Coherent State Path Integral}

The transition amplitude between an initial code state $\ket{\psi_i}$ and a final state $\ket{\psi_f}$ under noise evolution $e^{-iH_{\mathrm{noise}}t}$ can be formulated as a coherent state path integral over the Segal--Bargmann space. In the holomorphic representation, the bosonic Fock space is mapped to a space of analytic functions, and the path integral is constructed using the overcomplete basis of coherent states $\ket{\mathbf{z}}$. The transition amplitude is given by:
\begin{align}
\bra{\psi_f}e^{-iH_{\mathrm{noise}}t}\ket{\psi_i} &= \int \Dint\bar{\mathbf{z}}\Dint\mathbf{z} \, \mu(\mathbf{z},\bar{\mathbf{z}}) \notag \\
&\quad \times \exp\left(i\int_0^t \left[\frac{i}{2}(\bar{\mathbf{z}}\dot{\mathbf{z}} - \dot{\bar{\mathbf{z}}}\mathbf{z}) - H_{\mathrm{noise}}(\bar{\mathbf{z}}, \mathbf{z})\right]dt'\right).
\label{eq:pathint}
\end{align}
where the measure is normalized as
\begin{equation}
\Dint\bar{\mathbf{z}}\Dint\mathbf{z} \, \mu(\mathbf{z},\bar{\mathbf{z}}) = \lim_{N\to\infty} \prod_{k=1}^{N-1} \left(\frac{d^2\mathbf{z}_k}{\pi^{2n}} e^{-\|\mathbf{z}_k\|^2}\right),
\end{equation}
with $d^2\mathbf{z}_k = \prod_{j=1}^n d^2 z_{a_j}^{(k)} d^2 z_{b_j}^{(k)}$.

To properly interpret this functional integral, several crucial mathematical and physical details must be specified. First, the integration measure $\Dint\bar{\mathbf{z}}\Dint\mathbf{z}$ denotes integration over all continuous trajectories of the \textit{coherent state labels} $\mathbf{z}(t')$ and $\bar{\mathbf{z}}(t')$ in $\Cc^{2n}$. Crucially, the coherent state path integral requires \textit{mixed boundary conditions}: the holomorphic label $\mathbf{z}(t)$ is fixed at the final time $t$ to match the final state bra $\bra{\psi_f}$, while the anti-holomorphic label $\bar{\mathbf{z}}(0)$ is fixed at the initial time $0$ to match the initial state ket $\ket{\psi_i}$. This asymmetry is a fundamental feature of the overcomplete coherent state basis and ensures the correct resolution of the identity at the boundaries. The boundary terms contribute overlap factors $\braket{\psi_f}{\mathbf{z}(t)}$ and $\braket{\bar{\mathbf{z}}(0)}{\psi_i}$ that must be included in the full amplitude.

Second, the kinetic term in the action, $\frac{i}{2}(\bar{\mathbf{z}}\dot{\mathbf{z}} - \dot{\bar{\mathbf{z}}}\mathbf{z})$, is precisely the pullback of the canonical symplectic form $\omega = \frac{i}{2}\sum_j (d\bar{z}_{a_j} \wedge dz_{a_j} + d\bar{z}_{b_j} \wedge dz_{b_j})$ of the K\"ahler manifold to the time domain. This term generates the Berry phase (or geometric phase) accumulated along the trajectory in the projective Hilbert space. It ensures that the path integral respects the underlying unitary geometry of the quantum evolution.

Third, the noise Hamiltonian $H_{\mathrm{noise}}(\bar{\mathbf{z}}, \mathbf{z})$ must be understood as the normal symbol (or anti-normal symbol, depending on the chosen ordering) of the quantum Hamiltonian operator $\hat{H}_{\mathrm{noise}}$. In the Segal--Bargmann representation, creation and annihilation operators are replaced by the complex variables $z$ and derivatives $\partial_{\bar{z}}$ (or vice versa). For a path integral formulation, it is most convenient to use the normal symbol, where $\hat{a}^\dagger \mapsto \bar{z}$ and $\hat{a} \mapsto \partial_{\bar{z}}$, yielding a c-number function $H_{\mathrm{noise}}(\bar{\mathbf{z}}, \mathbf{z})$ that governs the classical-like dynamics of the coherent state trajectories. The noise model, typically a sum of local Pauli operators or Lindblad dissipators, translates into a specific polynomial or differential functional of the complex variables $\mathbf{z}$ and $\bar{\mathbf{z}}$.

\subsection{Error Correction as Constrained Path Integral}

The active process of quantum error correction---comprising syndrome extraction and recovery---fundamentally alters the unitary evolution of the system. In the path integral formalism, this non-unitary, projective process modifies the functional integral by inserting projections onto the syndrome sectors. The amplitude for a successfully corrected trajectory is expressed as a constrained path integral:
\begin{equation}
\mathcal{A}_{\mathrm{corrected}} = \int \Dint\bar{\mathbf{z}}\Dint\mathbf{z} \, \mu(\mathbf{z},\bar{\mathbf{z}}) \prod_{m=1}^{n-k}\delta(\hat{g}_m - 1) \exp\left(iS[\bar{\mathbf{z}}, \mathbf{z}]\right).
\label{eq:constrained_pathint}
\end{equation}
where $S[\bar{\mathbf{z}}, \mathbf{z}]$ is the action defined in Eq.~\eqref{eq:pathint}. The delta functions $\delta(\hat{g}_m - 1)$ enforce the stabilizer constraints, mathematically restricting the domain of the path integral strictly to the code submanifold $\Mcal_{\Ccal}$.

To give this expression rigorous mathematical meaning within the complex functional integral, the delta functions must be exponentiated using auxiliary fields. In the spirit of the Faddeev--Popov procedure for gauge theories, we introduce a set of real-valued Lagrange multiplier fields $\theta_m(t')$ for each stabilizer generator. The constraint delta function is rewritten as a functional Fourier transform:
\begin{equation}
\delta(g_m(\bar{\mathbf{z}}, \mathbf{z}) - 1) = \int \Dint\theta_m \exp\left( i \int_0^t dt' \, \theta_m(t') \left[ g_m(\bar{\mathbf{z}}, \mathbf{z}) - 1 \right] \right).
\end{equation}
Substituting this into Eq.~\eqref{eq:constrained_pathint}, the constrained path integral is transformed into an unconstrained integral over an extended phase space that includes both the physical variables $(\mathbf{z}, \bar{\mathbf{z}})$ and the auxiliary gauge fields $\theta_m$. The effective action becomes:
\begin{equation}
S_{\mathrm{eff}} = \int_0^t dt' \left[ \frac{i}{2}(\bar{\mathbf{z}}\dot{\mathbf{z}} - \dot{\bar{\mathbf{z}}}\mathbf{z}) - H_{\mathrm{noise}}(\bar{\mathbf{z}}, \mathbf{z}) + \sum_{m=1}^{n-k} \theta_m(t') (g_m(\bar{\mathbf{z}}, \mathbf{z}) - 1) \right].
\label{eq:extended_action}
\end{equation}

This extended action reveals the deep geometric structure of the error correction process. The auxiliary fields $\theta_m(t')$ act as dynamical gauge fields that enforce the trajectory to remain on the constraint surface $g_m = 1$. Geometrically, the conditions $g_m(\bar{\mathbf{z}}, \mathbf{z}) = 1$ define the code submanifold $\Mcal_{\Ccal}$ as a level set in the ambient K\"ahler space. The Lagrange multipliers $\theta_m$ generate flows along the normal directions to $\Mcal_{\Ccal}$, precisely counteracting the noise-induced drift that would otherwise push the state out of the code space. 

\begin{proposition}[Geometric Equivalence of Constrained Dynamics]
The constrained path integral \eqref{eq:constrained_pathint} is geometrically equivalent to the projection of the noise-evolved state onto the code submanifold $\Mcal_{\Ccal}$ along the normal bundle. In the semiclassical limit, the dominant contribution to the path integral comes from the classical trajectory that minimizes the extended action \eqref{eq:extended_action}. This classical trajectory corresponds to the optimal recovery path: the unique geodesic in the Fubini--Study metric that connects the noise-perturbed state back to the code submanifold $\Mcal_{\Ccal}$, orthogonal to the constraint surface.
\end{proposition}

Furthermore, this formulation naturally accommodates the discrete nature of actual quantum error correction protocols. In a physical architecture, syndrome measurements are not continuous but occur at discrete time intervals $\Delta t$. The path integral for a discrete error correction cycle is constructed by concatenating unconstrained noise evolution segments with instantaneous projection operators. Mathematically, this is represented by inserting the projection kernel $\Pi_{\Ccal}(\bar{\mathbf{z}}_f, \mathbf{z}_i) = \sum_{\mathbf{s}} \delta(\mathbf{s} - \mathbf{s}_{\mathrm{trivial}}) \dots$ at each measurement time slice. In the continuum limit of frequent measurements ($\Delta t \to 0$), the discrete projections converge to the continuous delta-function constraints of Eq.~\eqref{eq:constrained_pathint}, recovering the quantum Zeno effect where the system is dynamically frozen within the code submanifold by the continuous monitoring of the stabilizers.

\subsection{Continuous Syndrome Monitoring and the Effective Action}

In a realistic fault-tolerant architecture, syndrome extraction is not an instantaneous mathematical projection but a physical process involving continuous or repeated weak measurements. We can model the continuous monitoring of the stabilizer generators $\hat{g}_m$ by introducing a non-Hermitian penalty term to the classical action, which energetically suppresses trajectories that deviate from the code submanifold $\Mcal_{\Ccal}$. 

Let $g_m(\bar{\mathbf{z}}, \mathbf{z}) = \langle \mathbf{z} | \hat{g}_m | \mathbf{z} \rangle / \langle \mathbf{z} | \mathbf{z} \rangle$ be the coherent state expectation value of the $m$-th stabilizer. The continuous measurement of the syndrome with rate $\kappa$ modifies the effective action:
\begin{equation}
S_{\mathrm{eff}}[\bar{\mathbf{z}}, \mathbf{z}] = S[\bar{\mathbf{z}}, \mathbf{z}] - i\frac{\kappa}{2} \int_0^t dt' \sum_{m=1}^{n-k} \left| g_m(\bar{\mathbf{z}}, \mathbf{z}) - 1 \right|^2.
\label{eq:effective_action}
\end{equation}

The measurement rate $\kappa$ has dimensions of inverse time and characterizes the strength of the continuous monitoring. Physically, $\kappa^{-1}$ represents the characteristic timescale for syndrome extraction: for $\kappa \gg 1$, the measurement is strong and frequent, rapidly projecting the state onto the code submanifold; for $\kappa \sim 1$, the measurement is weak and continuous, allowing for gradual error detection. In superconducting qubit architectures, typical syndrome extraction times range from $100$ ns to $1$ $\mu$s, corresponding to $\kappa \sim 10^6$--$10^7$ s$^{-1}$ in natural units.

The imaginary penalty term acts as a dynamical quantum Zeno effect, confining the semiclassical trajectories to a tubular neighborhood of $\Mcal_{\Ccal}$. This formulation is rigorously justified by the theory of continuous quantum measurement (e.g., quantum trajectories and the stochastic master equation limit), where continuous monitoring of an observable $\hat{O}$ with measurement rate $\kappa$ leads to a non-Hermitian effective Hamiltonian $H_{\mathrm{eff}} = H - i\frac{\kappa}{2}\hat{O}^\dagger\hat{O}$ in the no-jump evolution. The penalty term in Eq.~\eqref{eq:effective_action} corresponds precisely to this no-jump dynamics for the stabilizer observables $\hat{g}_m - \I$, where the coherent-state expectation value $g_m(\bar{\mathbf{z}}, \mathbf{z}) - 1$ plays the role of the measurement signal. In the strong measurement limit $\kappa \to \infty$, the path integral localizes strictly onto the holomorphic constraint surface, recovering the exact constrained path integral of Eq.~\eqref{eq:constrained_pathint}.

\subsection{Semiclassical Approximation}

In the semiclassical limit, the path integral is dominated by classical trajectories satisfying:
\begin{equation}
i\dot{z}_j = \frac{\partial H_{\mathrm{noise}}}{\partial \bar{z}_j}, \qquad -i\dot{\bar{z}}_j = \frac{\partial H_{\mathrm{noise}}}{\partial z_j},
\label{eq:classical_eom}
\end{equation}
subject to the homogeneity constraints. The leading-order contribution yields:
\begin{equation}
\mathcal{A}_{\mathrm{sc}} = \Ncal\exp\left(\frac{i}{\hbar}S_{\mathrm{cl}}[\bar{\mathbf{z}}_{\mathrm{cl}}, \mathbf{z}_{\mathrm{cl}}]\right),
\label{eq:semiclassical}
\end{equation}
where $S_{\mathrm{cl}}$ is the action evaluated on the classical trajectory and $\Ncal$ is a prefactor from Gaussian fluctuations \cite{klauder1960,kirwin2007,ali2014}.

\subsection{Fluctuation Determinant and the Normal Bundle Geometry}

The prefactor $\Ncal$ in the semiclassical approximation (Eq.~\eqref{eq:semiclassical}) arises from the Gaussian integration over quadratic fluctuations around the classical trajectory. It is given by the Van Vleck--Pauli--Morette determinant:
\begin{equation}
\Ncal \propto \left[ \det \left( -\frac{\delta^2 S_{\mathrm{eff}}}{\delta \mathbf{z} \delta \bar{\mathbf{z}}} \right) \right]^{-1/2}.
\label{eq:van_vleck}
\end{equation}
The spectrum of the fluctuation operator decomposes naturally into longitudinal modes (tangent to $\Mcal_{\Ccal}$) and transverse modes (normal to $\Mcal_{\Ccal}$). 

The longitudinal fluctuations correspond to logical degrees of freedom and remain gapless (massless), reflecting the degeneracy of the code space. The transverse fluctuations, however, acquire a mass gap proportional to the measurement rate $\kappa$ and the extrinsic curvature (second fundamental form) of the embedding $\Mcal_{\Ccal} \hookrightarrow \CP^{2^n-1}$. Specifically, the effective mass matrix for the normal modes is:
\begin{equation}
(M^2)_{ab} = \kappa \sum_{m} \nabla_a (g_m - 1) \nabla_b (\bar{g}_m - 1) + \mathcal{K}_{ab},
\label{eq:mass_matrix}
\end{equation}
where $\nabla$ denotes the covariant derivative on the ambient space and $\mathcal{K}_{ab}$ is the extrinsic curvature tensor. A larger code distance $d$ geometrically corresponds to a steeper potential well (larger $\kappa$ or higher extrinsic curvature) separating the code submanifold from the images of uncorrectable errors, thereby exponentially suppressing the probability of logical transitions via tunneling in the path integral.

\subsection{Error Correction Fidelity}
For an initial pure code state $\ket{\psi_{\mathrm{code}}} \in \Mcal_{\Ccal}$, the entanglement fidelity of the error correction process is defined as:
\begin{equation}
\Fcal_{\mathrm{EC}} = \bra{\psi_{\mathrm{code}}} R \circ \Ecal (\ket{\psi_{\mathrm{code}}}\bra{\psi_{\mathrm{code}}}) \ket{\psi_{\mathrm{code}}},
\label{eq:fidelity}
\end{equation}
where $\Ecal$ is the quantum error channel representing the noise evolution, and $R$ is the completely positive trace-preserving (CPTP) recovery map. 

To analyze this within the holomorphic framework, we must express both $\Ecal$ and $R$ in the Segal--Bargmann space. The error channel is given by its Kraus representation $\Ecal(\rho) = \sum_\mu K_\mu \rho K_\mu^\dagger$. In the holomorphic picture, each Kraus operator $K_\mu$ maps to a differential or multiplication operator $\hat{K}_\mu$ acting on the analytic wavefunction $f(\mathbf{z})$. Physically, the action of $\hat{K}_\mu$ displaces the state from the code submanifold $\Mcal_{\Ccal}$ into the ambient projective space, typically pushing it into the normal bundle $N\Mcal_{\Ccal}$. 

The recovery map $R$ is similarly decomposed into recovery operators $R_\nu$, corresponding to the measurement of a specific syndrome $\mathbf{s}_\nu$ and the application of a corrective unitary. In the holomorphic representation, the syndrome measurement is a spectral projection $\hat{\Pi}_{\mathbf{s}_\nu}$ onto the eigenspaces of the commuting stabilizer operators $\hat{g}_m$. The corrective unitary is represented by a holomorphic differential operator $\hat{U}_{\mathbf{s}_\nu}$ designed to map the error-translated submanifold back to $\Mcal_{\Ccal}$. Thus, the recovery operators take the form $\hat{R}_\nu = \hat{U}_{\mathbf{s}_\nu} \hat{\Pi}_{\mathbf{s}_\nu}$. The fidelity \eqref{eq:fidelity} then becomes a sum of overlaps between the initial holomorphic wavefunction and the recovered wavefunctions across all syndrome sectors.

\begin{remark}[Simple Example: Single-Qubit Depolarizing Channel]
For a single physical qubit undergoing a depolarizing channel with error probability $p$, the Kraus operators are $K_0 = \sqrt{1-3p/4}\I$ and $K_{x,y,z} = \sqrt{p/4}\sigma_{x,y,z}$. In the holomorphic path integral, the classical trajectory for the no-error case ($K_0$) remains on $\Mcal_{\Ccal}$ with action $S_0 \approx 0$. The error trajectories ($K_{x,y,z}$) correspond to discrete geodesic jumps of Fubini--Study length $\pi/2$ away from $\Mcal_{\Ccal}$, each contributing an action penalty of $\ln(4/3p)$. The semiclassical fidelity is thus $\Fcal_{\mathrm{EC}}^{\mathrm{sc}} \approx (1-3p/4) + O(p^2)$, perfectly matching the leading-order algebraic result.
\end{remark}

In the semiclassical limit, where the noise is weak and the path integral is dominated by classical trajectories, the fidelity admits a systematic expansion. To leading order, the fidelity is determined by the probability of the error channel producing an error that exceeds the correction capacity of the code:
\begin{equation}
\Fcal_{\mathrm{EC}} \approx 1 - \sum_{\mu:\, \wt(E_\mu)>t} p_\mu + O(\hbar),
\label{eq:fidelity_sc}
\end{equation}
where $p_\mu = \mathrm{Tr}(K_\mu \rho K_\mu^\dagger)$ is the probability of the error $E_\mu$ occurring, and $t = \lfloor(d-1)/2\rfloor$ is the maximum correctable weight. The leading correction to unity is simply the total probability weight of all uncorrectable error chains (those with weight strictly greater than $t$), which cause the state to cross into the wrong syndrome sector and result in a logical failure.

The $O(\hbar)$ corrections in Eq.~\eqref{eq:fidelity_sc} arise from quantum fluctuations around the classical recovery trajectory in the path integral formulation. These corrections are computable via a systematic loop expansion. The 1-loop correction is governed by the Van Vleck--Pauli--Morette determinant, which accounts for the focusing or defocusing of the probability flow in the normal bundle due to the extrinsic curvature of the code submanifold $\Mcal_{\Ccal}$. Physically, these loop corrections capture multi-error interference effects, correlated noise processes, and the finite-width of the wavepacket in the Segal--Bargmann space, which are entirely missed by the classical (leading-order) geometric picture.

\begin{proposition}[Geometric Interpretation of Infidelity]
The infidelity $1 - \Fcal_{\mathrm{EC}}$ is geometrically bounded by the Fubini--Study distance between the initial state and the recovered state. Specifically, if the error and recovery process results in a final state $\ket{\psi_f}$, the infidelity is related to the geodesic distance in the projective Hilbert space by:
\begin{equation}
1 - \Fcal_{\mathrm{EC}} = \sin^2 \left( D_{\mathrm{FS}}([\psi_{\mathrm{code}}], [\psi_f]) \right) \approx D_{\mathrm{FS}}^2([\psi_{\mathrm{code}}], [\psi_f]).
\label{eq:infidelity_distance}
\end{equation}
For a successful recovery, the final state $[\psi_f]$ lies on the code submanifold $\Mcal_{\Ccal}$. The residual infidelity arises because the recovery map $R$ may not perfectly invert the specific continuous deformation induced by the Kraus operator $\hat{K}_\mu$, leaving a residual tangent vector in $T_{[\psi]}\Mcal_{\Ccal}$ (a coherent logical error) or a residual normal vector in $N_{[\psi]}\Mcal_{\Ccal}$ (an uncorrected physical error).
\end{proposition}

This geometric picture of fidelity provides a rigorous foundation for the quantum fault-tolerance threshold theorem, elevating it from a purely algebraic counting argument to a profound statement about the topology of the ambient K\"ahler geometry. To formalize this, we must express the probability of a logical failure as a path integral over trajectories in the projective Hilbert space $\CP^{2^n-1}$ that connect the initial code state to a logically corrupted state.

Let $\Gamma_{\mathrm{logical}}$ denote the space of all continuous paths $\gamma(\tau)$ in $\CP^{2^n-1}$ that start at the code submanifold $\Mcal_{\Ccal}$ at $\tau=0$ and end at a logically flipped state $\bar{L}\ket{\psi_{\mathrm{code}}}$ at $\tau=t$, where $\bar{L}$ is a logical operator. The infidelity, which is the probability of a logical fault, is given by the path integral over all such failure trajectories, weighted by their semiclassical probability amplitude:
\begin{equation}
1 - \Fcal_{\mathrm{EC}} = \int_{\gamma \in \Gamma_{\mathrm{logical}}} \Dint\gamma \, \exp\left( - \frac{1}{\hbar} S_{\mathrm{eff}}[\gamma] \right),
\label{eq:logical_path_integral}
\end{equation}
where $S_{\mathrm{eff}}[\gamma]$ is the effective Euclidean action of the noise trajectory $\gamma$. In the Segal--Bargmann space, this action is determined by the noise Hamiltonian and the Fubini--Study metric $g_{\FS}$ of the ambient space. For a local depolarizing or Pauli noise model with physical error rate $p$, the action of a trajectory is directly related to the geometric length of the path and the strength of the noise. 

Specifically, a discrete Pauli error of weight $w$ corresponds to a path $\gamma$ composed of $w$ discrete geodesic segments in $\CP^{2^n-1}$, each representing a single-qubit error. The effective action for such a path is additive and scales with the weight:
\begin{equation}
S_{\mathrm{eff}}[\gamma_w] = w \ln\left(\frac{1}{p}\right) + \mathcal{O}(1),
\label{eq:action_weight}
\end{equation}
where the $\mathcal{O}(1)$ term accounts for the geometric prefactors and the Fubini--Study length of the individual single-qubit rotations. This equation establishes the crucial dictionary between the algebraic weight of an error and the geometric ``energy barrier'' (action) it must overcome in the ambient space.

The code distance $d$ and the correction capacity $t = \lfloor(d-1)/2\rfloor$ dictate the topology of the allowed recovery paths. The recovery map $R$ acts as a geometric projection that successfully returns any trajectory to $\Mcal_{\Ccal}$ provided the path does not cross the ``midpoint'' between distinct logical sectors. Consequently, any trajectory $\gamma \in \Gamma_{\mathrm{logical}}$ that results in an uncorrectable logical fault must correspond to an error chain of weight $w \geq t+1$. 

In the semiclassical limit ($\hbar \to 0$ or equivalently $p \to 0$), the path integral \eqref{eq:logical_path_integral} is dominated by the classical trajectory of least action---the ``error instanton'' $\gamma_{\mathrm{inst}}$---which minimizes $S_{\mathrm{eff}}$ subject to the boundary conditions of a logical fault. The minimal action is therefore achieved by an error chain of the minimum uncorrectable weight, $w_{\min} = t+1$:
\begin{equation}
S_{\mathrm{inst}} = \min_{\gamma \in \Gamma_{\mathrm{logical}}} S_{\mathrm{eff}}[\gamma] = (t+1) \ln\left(\frac{1}{p}\right) + \mathcal{O}(1).
\label{eq:instanton_action}
\end{equation}

Evaluating the path integral via the saddle-point approximation yields the infidelity in terms of the instanton action and the 1-loop fluctuation determinant $\mathcal{A}$ (which accounts for Gaussian quantum fluctuations around the classical instanton path):
\begin{equation}
1 - \Fcal_{\mathrm{EC}} \approx \mathcal{A} \exp\left( - S_{\mathrm{inst}} \right) = \mathcal{A} \exp\left( - (t+1) \ln\left(\frac{1}{p}\right) \right) = \mathcal{A} \, p^{t+1}.
\label{eq:infidelity_scaling}
\end{equation}
This rigorously recovers the $O(p^{t+1})$ scaling of the infidelity. The threshold theorem thus emerges naturally: as long as the physical noise strength $p$ is sufficiently small, the action barrier $S_{\mathrm{inst}}$ grows linearly with the code capacity $t$, exponentially suppressing the probability of logical failure.

Ultimately, the holomorphic path integral transforms the algebraic threshold theorem into a statement about the topological suppression of tunneling events. The code submanifold $\Mcal_{\Ccal}$ and its logically flipped images are separated by a massive Fubini--Study ``energy barrier'' in the ambient K\"ahler geometry. The noise drives the system via quantum tunneling (instantons) through this barrier. The fault-tolerance threshold is precisely the critical noise strength below which the tunneling rate between disconnected logical sectors is topologically suppressed by the exponential of the minimum geodesic action, ensuring the stable preservation of quantum information.

\subsection{Geometric Phase Accumulation and Recovery Holonomy}

The error and subsequent recovery process traces out a closed loop in the projective Hilbert space $\CP^{2^n-1}$. An initial code state $[\psi_i]$ undergoes a noise-induced excursion into the normal bundle $N\Mcal_{\Ccal}$, reaching an errored state $[\psi_E]$, before the recovery map $\mathcal{R}$ projects it back to a final state $[\psi_f] \in \Mcal_{\Ccal}$. For successful error correction, $[\psi_f] = [\psi_i]$ up to a global phase.

The total phase accumulated along this cycle $\Gamma$ is given by the path integral weight:
\begin{equation}
\Phi_{\mathrm{total}} = \oint_{\Gamma} \mathcal{A} = \oint_{\Gamma} \frac{i}{2} \frac{\bar{\mathbf{z}} d\mathbf{z} - d\bar{\mathbf{z}} \mathbf{z}}{\|\mathbf{z}\|^2}.
\label{eq:total_phase}
\end{equation}
By Stokes' theorem, this can be expressed as the integral of the Fubini--Study symplectic form $\omega_{\FS}$ over a surface $\Sigma$ bounded by $\Gamma$:
\begin{equation}
\Phi_{\mathrm{geom}} = \int_{\Sigma} \omega_{\FS}.
\label{eq:geom_phase_stokes}
\end{equation}
For a fault-tolerant recovery protocol, the geometric phase $\Phi_{\mathrm{geom}}$ must be independent of the specific logical state $[\psi_i] \in \Mcal_{\Ccal}$; otherwise, the recovery operation would induce a state-dependent logical phase gate, resulting in coherent logical errors. This topological requirement constrains the curvature of the recovery map's connection over the code submanifold, ensuring that the holonomy of the error-recovery cycle is strictly an element of the global $U(1)$ fiber rather than a non-trivial logical operation.

\subsection{Weakly-Entangled Regime}

For codes whose codewords remain near the Segre variety during evolution (weakly entangled), the semiclassical approximation is particularly accurate:

\begin{proposition}
If the code state $[\psi]$ satisfies $D_{\FS}([\psi], \Sigma_n) \ll 1$ throughout the error correction cycle, then the semiclassical error correction fidelity satisfies:
\begin{equation}
|\Fcal_{\mathrm{EC}} - \Fcal_{\mathrm{EC}}^{\mathrm{sc}}| = O(D_{\FS}([\psi], \Sigma_n)^2).
\label{eq:sc_accuracy}
\end{equation}
\end{proposition}

\begin{proof}
The semiclassical approximation of the path integral relies on the stationary phase approximation, where the dominant contribution comes from the classical trajectory, and quantum fluctuations are treated as Gaussian integrals around this trajectory.
The magnitude of these quantum fluctuations is governed by the second variation of the action, which is related to the curvature of the manifold. Near the Segre variety $\Sigma_n$, the state is approximately separable, meaning the entanglement is low and the geometry of the state space is well-approximated by a flat, product space.
Let $D = D_{\FS}([\psi], \Sigma_n)$ be the Fubini--Study distance to the nearest separable state. The deviation of the true metric from the flat metric scales as $O(D^2)$. Consequently, the higher-order terms in the expansion of the action (cubic and higher in the fluctuations) are suppressed by factors of $D$.
The error in the semiclassical fidelity $\Fcal_{\mathrm{EC}}^{\mathrm{sc}}$ compared to the exact fidelity $\Fcal_{\mathrm{EC}}$ arises from neglecting these higher-order fluctuation terms (loop corrections). Since the leading correction to the Gaussian integral comes from the cubic terms in the action, and these terms are proportional to the curvature (which is $O(D^2)$), the error in the fidelity is bounded by $O(D^2)$. Thus, $|\Fcal_{\mathrm{EC}} - \Fcal_{\mathrm{EC}}^{\mathrm{sc}}| = O(D_{\FS}([\psi], \Sigma_n)^2)$.
\end{proof}

This is because quantum fluctuations (loop corrections) are suppressed near the Segre variety, where the state is approximately separable and classical-like.

\section{Extensions and Applications}
\label{sec:extensions}

\subsection{Topological Codes: Surface Code in Holomorphic Representation}

The toric code \cite{kitaev2003} on a genus-$g$ surface has stabilizer generators corresponding to star operators $A_v$ and plaquette operators $B_p$. In the holomorphic representation:
\begin{align}
\hat{A}_v &= \prod_{j\in\mathrm{star}(v)}(z_{a_j}\partial_{z_{b_j}} + z_{b_j}\partial_{z_{a_j}}), \label{eq:star}\\
\hat{B}_p &= \prod_{j\in\partial p}(z_{a_j}\partial_{z_{a_j}} - z_{b_j}\partial_{z_{b_j}}). \label{eq:plaquette}
\end{align}

The topological degeneracy of the code (dimension $4^g$ for genus $g$) corresponds to the $4^g$ distinct winding number sectors on the torus $\T^{2n}$ that are not connected by local stabilizer operations. Anyons correspond to violations of the stabilizer constraints at individual vertices or plaquettes, manifesting as localized defects in the holomorphic function.

\subsection{Subsystem Codes}

For subsystem codes \cite{bacon2006}, the gauge group $\mathcal{G}$ contains the stabilizer $\Scal$ as a subgroup, with the gauge qubits encoding redundant degrees of freedom that need not be protected. In the holomorphic representation, the gauge degrees of freedom correspond to directions in $\HSB$ along which the code submanifold $\Mcal_{\Ccal}$ is fibered. Specifically, the gauge group action defines a foliation of $\Mcal_{\Ccal}$ into gauge orbits, and the physical logical information resides in the leaf space of this foliation:
\begin{equation}
\mathcal{H}_{\mathrm{logical}} = \Mcal_{\Ccal} / \mathcal{G}_{\mathrm{gauge}},
\label{eq:subsystem}
\end{equation}
which geometrically corresponds to a principal bundle reduction or fiber bundle quotient of the code submanifold. The gauge generators act as holomorphic vector fields tangent to the fibers, and syndrome extraction need only detect errors modulo the gauge group, reducing the number of required stabilizer measurements. This fiber bundle perspective clarifies why subsystem codes can achieve the same error correction capability with fewer syndrome measurements compared to stabilizer codes encoding the same number of logical qubits.

\subsection{Continuous-Variable Error Correction}

The holomorphic framework naturally extends to continuous-variable (CV) quantum error correction \cite{braunstein1998,gkp2001}. In the CV setting, the homogeneity constraint is relaxed to allow arbitrary-degree holomorphic functions, and the code space is defined by a different set of constraints (e.g., eigenstates of $\hat{x} - \hat{p}$ for GKP codes). The K\"ahler geometry and fiber bundle structure persist, providing a unified geometric description of both discrete-variable and CV error correction. Furthermore, the characterization of these non-Gaussian CV states and the extraction of their expansion coefficients can be practically implemented via CV quantum algorithms operating natively in the Segal--Bargmann space \cite{almasri2026projection}.

\subsection{Fault-Tolerant Gate Implementation}

In the holomorphic representation, fault-tolerant gates are those that preserve the code submanifold $\Mcal_{\Ccal}$ while implementing the desired logical operation. The transversal implementation of gates corresponds to differential operators that act independently on each qubit's holomorphic variables:
\begin{equation}
\hat{U}_{\mathrm{transversal}} = \prod_{j=1}^n \hat{U}_j,
\label{eq:transversal}
\end{equation}
which geometrically corresponds to a product diffeomorphism on $\T^{2n}$ that preserves the code submanifold.

\section{Conclusion}
\label{sec:conclusion}

We have established a holomorphic framework for quantum error correction codes within the Segal--Bargmann space. By encoding qubits via the Schwinger boson representation and mapping to holomorphic functions, we have demonstrated that the entire architecture of stabilizer quantum error correction---code spaces, error operators, syndrome extraction, and recovery---admits a natural geometric interpretation. In this framework, stabilizer generators become commuting holomorphic differential operators whose joint eigenspaces define the code manifold; errors manifest as perturbations of the homogeneity constraint that displace the state into the normal bundle; syndrome extraction corresponds to spectral projection onto eigenspaces of these operators; and recovery acts as a Hamiltonian flow that restores the original topological sector. On the restricted torus $\T^{2n}$, this dictionary acquires a particularly transparent form: error syndromes are encoded as discrete winding number shifts in $\Z^{2n}$, providing a topological classification of errors through integer-valued invariants and a natural mechanism for detection via discrete jumps in the homotopy class of the trajectory.

Beyond the toroidal geometry, the full K\"ahler structure of the ambient projective space $\CP^{2^n-1}$ endows the correctability condition with a clean geometric meaning. The Knill--Laflamme condition translates into the mutual orthogonality of error-translated code submanifolds with respect to the Fubini--Study metric, while the code distance corresponds to the minimum geodesic separation between the code manifold and its images under weight-$d$ logical operators. The $U(1)^n$ fiber bundle structure inherent in the Schwinger encoding further reveals a topological protection mechanism: global phase noise acts as vertical automorphisms of the bundle and is automatically unobservable upon projectivization, while local Pauli errors generate horizontal displacements on the Bloch sphere base space that require active syndrome extraction. The K\"ahler potential also enables a coherent-state path-integral formulation of error correction dynamics, in which the semiclassical approximation becomes exact in the weakly-entangled regime near the Segre variety, with the fluctuation determinant encoding the extrinsic curvature of the code submanifold.

Furthermore, while the holomorphic representation provides profound geometric insights, it is worth noting that the computational complexity of evaluating high-order differential operators scales polynomially with the number of qubits for sparse error models, offering a complementary perspective to the exponential scaling of matrix representations in the standard computational basis.

Looking forward, several directions emerge naturally from this geometric unification. These include extending the framework to non-stabilizer codes and approximate error correction, developing numerical algorithms based on the holomorphic representation for simulating large-scale fault-tolerant protocols, exploring connections to holomorphic vector bundles and algebraic geometry for topological code classification, investigating the role of K\"ahler quantization in constructing new code families, and applying the geometric framework to derive refined bounds on threshold theorems. Ultimately, the holomorphic perspective reveals that quantum error correction is fundamentally a geometric operation: encoding maps information into a protected submanifold, errors deform the state away from this submanifold, syndrome measurement identifies the direction of deformation, and recovery projects back along the normal bundle. This unification of algebraic, geometric, and topological viewpoints provides both conceptual clarity and new computational tools for the design and analysis of quantum error correction protocols.



\end{document}